\documentclass[aoas,preprint]{imsart_v2}
\pdfoutput=1
\usepackage{amssymb, graphicx,float}
\usepackage{amsthm,amsmath,xcolor,epstopdf,multirow,subcaption}
\RequirePackage[colorlinks,citecolor=blue,urlcolor=blue]{hyperref}
\RequirePackage{natbib}
\numberwithin{equation}{section}
\startlocaldefs
\definecolor{DarkRed}{rgb}{.7,0,.4}

\setattribute{journal}{name}{}

\def\ci{\cite}
\def\cp{\citep}
\def\mc{\mathcal}
\def\bco{\iffalse}

\newcommand{\bea}{\begin{eqnarray*}}
\newcommand{\eea}{\end{eqnarray*}}

\newcommand{\ed}{\end{document}}

\newcommand{\btab}{\begin{tabular}}
\newcommand{\etab}{\end{tabular}}
\newcommand{\bc}{\begin{center}}
\newcommand{\ec}{\end{center}}

\newcommand{\bi}{\begin{itemize}}
\newcommand{\ei}{\end{itemize}}
\newcommand{\bfi}{\begin{figure}}
\newcommand{\efi}{\end{figure}}
\newcommand{\ben}{\begin{enumerate}}
\newcommand{\een}{\end{enumerate}}
\newcommand{\bdes}{\begin{description}}
\newcommand{\edes}{\end{description}}
\newcommand{\bay}{\begin{array}}
\newcommand{\eay}{\end{array}}

\newtheorem{Theorem}{Theorem}%[section]

\newtheorem{Corollary}{Corollary}

\newtheorem{Remark}{Remark}

\newcommand{\be}{\begin{eqnarray}}
\newcommand{\ee}{\end{eqnarray}}

\def\sn{\sum_{i=1}^n}

\def\om{\omega}
\def\Var{{\rm Var}}
\def\Cov{{\rm Cov}}

\DeclareMathOperator*{\argmin}{argmin}
\def\Om{\Omega}
\def\om{\omega}

\def\inv{^{-1}}

\def\hmu{\hat{\mu}}
\def\SNW{\hat{\Sigma}^{\textrm{NW}}}
\def\SLF{\hat{\Sigma}^{\textrm{LF}}}
\def\SLFt{\tilde{\Sigma}^{\textrm{LF}}}
\def\Sig{\Sigma}
\def\hSig{\hat{\Sig}}
\def\mcSp{\mathcal{S}_p}
\def\bco{\iffalse}
\def\d{\textrm{d}}

\begin{document}

\begin{frontmatter}

% "Title of the paper"
\title{Fr\'echet Estimation of Time-Varying Covariance Matrices From Sparse Data, With Application to the Regional Co-Evolution of Myelination in the Developing Brain}
\runtitle{Time-Varying Covariance from Sparse Data}

\author{\fnms{Alexander} \snm{Petersen,}\thanksref{m1}\corref{}\ead[label=e1]{petersen@pstat.ucsb.edu}},
\author{\fnms{Sean} \snm{Deoni}\thanksref{m2,t1}\ead[label=e2]{hgmueller@ucdavis.edu}}
\and
\author{\fnms{Hans-Georg} \snm{M\"uller}\thanksref{m3,t2}\ead[label=e2]{hgmueller@ucdavis.edu}}
\affiliation{Department of Statistics and Applied Probability, University of California, Santa Barbara\thanksmark{m1}}
\affiliation{Department of Pediatrics, Brown University, Providence\thanksmark{m2}}
\affiliation{Department of Statistics, University of California, Davis\thanksmark{m3}}
\thankstext{t1}{Supported by the National Institutes of Mental Health (R01 MH087510) and the Bill and Melinda Gates Foundation (OPP11002016).} \thankstext{t2}{Supported by NSF grants DMS-1407852 and DMS-1712864 and the Bill and Melinda Gates Foundation (OPP1119700).}
\address{Address of Alexander Petersen\\
Statistics and Applied Probability \\
University of California \\
Santa Barbara, CA 93106-3110 \\
\printead{e1}}
%\address{Address of Hans-Georg M\"uller\\
%Department of Statistics \\
%Mathematical Sciences Building 4118 \\
%399 Crocker Lane \\
%University of California, Davis \\
%One Shields Avenue \\
%Davis, CA 95616 \\
%\printead{e2}}


\runauthor{Petersen, Deoni and M\"uller}

\begin{abstract}

Assessing brain development for small infants is important for determining how the human brain grows during the early period of life when the rate of brain growth is at its peak. The development of MRI techniques 
has enabled the quantification of brain development. A  key quantity that can be extracted from MRI measurements is the level of myelination, where myelin acts as an insulator around nerve fibers and 
its deployment makes nerve pulse propagation more efficient. The co-variation of myelin deployment across different brain regions provides insights  into the co-development of brain regions and can be assessed as a correlation matrix that varies with age. Typically, available data for each child are very sparse, due to the cost and logistic difficulties of arranging MRI brain scans for infants. We showcase here a method where data per subject are limited to measurements taken at only one random age, 
while aiming at the time-varying dynamics. This situation is encountered more generally in cross-sectional studies where one observes  $p$-dimensional vectors at one random time point per subject and is interested in the $p \times p$ correlation matrix function over the time domain. The challenge is that at each observation time one observes only a $p$-vector of measurements but not a covariance or correlation matrix. For such very sparse data, we develop a Fr\'echet estimation method. Given a metric on the space of covariance matrices, the proposed method generates a matrix function where at each time the matrix is a non-negative definite covariance matrix, for which we demonstrate consistency properties. We discuss how this approach can be applied to myelin data in the developing brain and what insights can be gained.

\end{abstract}

\end{frontmatter}

\section{Introduction}

Modern Magnetic Resonance Imaging (MRI) methodology has made it possible to study structural elements of the brain for small 
infants. Of special interest is to utilize this technique for the study of how the brain grows and develops in the early childhood years, a period of rapid brain and cognitive growth 
\cp{lenr:06}, during which 
the brain's structural and functional networks are formed.  The maturity of the developing brain is customarily measured in terms of level of myelination, which can be extracted from MRI signals. Myelin serves as an insulating sheath around nerve fibers and higher myelin levels are thought to be associated with more efficient signaling and thus improved brain function.   Altered brain connectivity has been hypothesized to underlie a number of intellectual, behavioral, and psychiatric disorders  \cp{d04, d16, d28}. To evaluate the development of brain connectivity, the  evolution  of correlations  between myelin levels for various white matter  brain regions as a function of age in normal children is of interest.   This is because brain development can be characterized by the degree of myelination and time-varying correlation matrices that quantify the dynamic correlations of myelination between different brain regions 
pinpoint the regions that are developing at similar age periods, characterizing the spatio-temporal brain development map. 

A general problem for the analysis of MRI studies in children  is the sparsity of data in the time domain, which is a consequence of cost and logistic difficulties. Even in studies that were intended to be longitudinal, for many infants only measurements at one  point in time are available. 
This motivates the development of the approach that we present here, where we obtain consistent estimates of the underlying  time-varying correlation matrices from sparse observations, where the estimates are guaranteed to be 
symmetric and nonnegative definite matrices at each time argument.  
We consider here the sparsest possible case, namely that  each subject is observed at only one random age. Our methodology is  of interest for any study where one aims at time-dynamic information for the  correlation/covariance  structure of continuous multivariate observations, when one essentially has only one measurement per subject  available. 

In our application to brain development, the underlying process of each subject corresponds to the myelination levels for $p=21$ brain regions, which evolve continuously over time.  However, observations are measured only at a single time point, so that the data available for each subject are the (random) time point at which the measurement was taken as well as a multivariate $p$-vector of region-specific myelin levels. As the data are vectors from which the smoothly varying correlation matrix is the desired target, we show here that an adaptation of Fr\'echet regression \cp{mull:18:3} provides a simple and effective approach to recover the correlation matrices. The proposed method produces for each age an estimate of the population correlation matrix with the following properties: the estimates are correlation matrices themselves; and they are consistent, with convergence rates that can be derived from the general theory of Fr\'echet regression. 

In Section~\ref{sec: bambam}, we will demonstrate these methods with data obtained as part of an on-going accelerated longitudinal neuroimaging study, the Brown University Assessment of Myelination and Behaviour Across Maturation (BAMBAM) study \cp{d143,d165}, by estimating time-varying correlation matrix functions. The cross-correlation matrices of myelination level in dependency on age entail a spatio-temporal map of development that indicates the locations of early and late development, as well as the changing connectivity of various sets of regions of interest. For the  children in the BAMBAM study, at each age where a MRI was obtained, the Early Learning Composite (ELC) score is also available, which quantifies neurocognitive development. As an application of the estimated covariance structure between the myelination levels of various brain regions across age, we use this estimated covariance matrix, in addition to a similarly constructed cross-correlation between the ELC outcome and the brain region myelin levels, to obtain linear regressions of ELC versus the actual myelination levels for the $p=21$ brain regions. Since the coefficients in the linear regression are also time-dependent, without the age-dependent covariance and cross-covariance estimates this regression would not be possible as for each age where predictors (brain myelination levels) and response (ELC) are available, we have only one measurement.  Our analysis pinpoints the relevance of various brain regions for neurocognitive outcomes in dependence of the age of infants. 

Previous work on modeling and quantifying brain development has not addressed the modeling of cross-correlations and covariances, with the exception of \ci{mull:17:10},
where ad hoc kernel smoothing techniques were employed to obtain pairwise cross-correlation functions for the sparsely observed BAMBAM data. However, the ensemble of the recovered pairwise correlations does not form a valid covariance/correlation matrix function, which at each age argument is required to be symmetric and non-negative definite \cp{mull:16:2}. Of these properties, symmetry can be easily guaranteed, while non-negative definiteness is a different matter and is not guaranteed for pairwise correlation function estimation approaches.  

In this paper, we study kernel smoothing methods for the coherent estimation of these time-varying covariance matrices, where ``coherent" means that the estimators are symmetric nonnegative definite matrices at all time points.  In Section~\ref{ss: mfd_methods}, we describe existing smoothing techniques which allow one to estimate the individual elements of each covariance matrix and discuss their drawbacks in the current setting.  Sections~\ref{ss: cfm} and \ref{ss: frechet_est} reformulate the problem as the smoothing of entire covariance matrices, opening a connection to a recent method for regression for complex objects, the  \emph{Fr\'echet regression} model.  Section~\ref{sec: bambam} gives more detailed background on the BAMBAM study and the Fr\'echet estimation technique is applied to study dynamic myelin correlations.  We also demonstrate that the estimated covariance matrix functions are instrumental to establish the relation between the myelination levels over the 21 brain regions with concurrent ELC scores.  Section~\ref{sec: sims} validates the approach by demonstrating in a simulation study the ability of the Fr\'echet estimator to consistently recover dynamic covariances from such sparse data using similar sample size and dimension, and its superiority to an alternative kernel estimator.  Section~\ref{sec: theory} demonstrates theoretical advantages of the Fr\'echet smoother and extends theory for consistent estimation, which, by necessity, must be performed using biased and degenerate covariance matrices as only multivariate random vectors are observed.

\section{Regression with Covariance as Response}
\label{sec: methods}

\subsection{Preliminaries}
\label{ss: prelim}

We begin by defining  the data objects and targets for estimation.  The observed data consist of i.i.d. pairs $(X_i, Y_i)$, $i = 1,\ldots, n$, where $X_i$ represents the age at which the child was examined and $Y_i\in \mathbb{R}^{p}$ is a vector that contains the myelin water fraction (MWF) levels at $p = 21$ brain regions measured during the visit (only available at one age).   It is natural to think of $Y_i$ as a point observation of a latent smooth multivariate stochastic process \mbox{$\{U_i(x) \in \mathbb{R}^p:\, x \in [0, T]\}$,} where the $U_i$ are i.i.d. mean-square continuous processes \cp{ash:75,hsin:15}, and \mbox{$Y_i = U_i(X_i).$}  Thus, the process $U_i$ tracks the MWF levels of subject $i$ at all ages, and $Y_i$ is the value of this process at a random age $X_i$ that is independent of the process $U_i.$ The targets of interest are the cross-sectional mean and covariance of the latent process, which can be expressed in terms of the conditional mean and covariance of $Y_i | X_i$ via
\begin{equation}
\label{eq: mean_cov}
\mu(x) = E(Y_i | X_i = x) = E(U_i(x)), \quad \Sigma(x) = \Cov(Y_i |X_i = x) = \Cov(U_i(x)).
\end{equation}

Notice that the covariance target is a function $\Sigma:[0,T] \rightarrow \mcSp$, where $\mcSp$ is the space of symmetric nonnegative definite matrices of dimension $p$. 
Similarly, one is interested in the corresponding pointwise correlation matrices $\mc{R}(x)$.  Estimators $\hmu_j(x)$ of the component functions $\mu_j(x)$ of $\mu(x)$ can be obtained in a variety of ways, including smoothing splines and local polynomial methods or other scatterplot smoothers, and have been well explored while this is not the case for the estimation of $\Sigma(x)$. Therefore we will focus our attention on the covariance estimation. 

\subsection{Cross-Covariance Estimation}
\label{ss: mfd_methods}

Established nonparametric techniques could be invoked for the estimation of the individual matrix elements $\Sigma_{jk}(x)$ by smoothing so-called \emph{raw covariances}\begin{equation}
\label{eq: raw_cov}
c_{ijk} = (Y_{ij} - \hmu_j(X_{i}))(Y_{ik} - \hmu_{k}(X_{i})).
\end{equation}
Let $K$ be a univariate density function, $h_{jk}$  a bandwidth and $K_h(y) = h\inv K(yh\inv)$.  Local constant (Nadaraya-Watson) and local linear smoothing are two common tools that can be enlisted to smooth these raw covariances, with respective estimators $\hat{\Sigma}_{jk}^0(x) = \hat{a}$ and $\hat{\Sigma}_{jk}^1(x) = \hat{b}$, where
\begin{align*}
\hat{a} &= \argmin_{a \in \mathbb{R}} \sn K_{h_{jk}}(X_i - x)(c_{ijk} - a)^2, \\ (\hat{b},\hat{c}) &= \argmin_{b,c\in \mathbb{R}}\sn K_{h_{jk}}(X_i - x)(c_{ijk} - b - c(X_i - x))^2.
\end{align*}
As solutions to weighted least squares problems, both of these kernel estimators take the form of a weighted average \cp{fan:96},
\begin{equation}
\label{eq: kern_smooth}
\hSig_{jk}^m(x) = \sn w_{in,m}c_{ijk},
\end{equation}
where $m = 0,1$ correspond to Nadaraya-Watson and local linear smoothing, respectively.  Letting $$r_{\ell} = \frac{1}{n}\sn K_{h_{jk}}(X_{i} - x)(X_{i} - x)^\ell, \quad \ell = 0,1,2,$$ and $\sigma^2 = r_0r_2-r_1^2$, the weights $w_{in,m}$ in (\ref{eq: kern_smooth}) are, respectively, 
\begin{align}
w_{in,0} &= w_{in,0}(x, h_{jk}) = \frac{K_{h_{jk}}(X_{i} - x)}{\sn K_{h_{jk}}(X_{i} - x)} \quad \textrm{and} \label{w0}\\
w_{in,1} &= w_{in,1}(x, h_{jk}) = \frac{1}{\sigma^2}K_{h_{jk}}(X_{i} - x)\left[r_2 - r_1(X_{i} - x)\right], \label{eq: weights}
\end{align}  
both satisfying  $\sn w_{in,m} = 1.$  This approach is closely related to the method used in \ci{mull:17:10}.

The above described  estimators are known to be consistent on the interior $(0, T)$, while the local linear estimator is preferable due to its improved performance near the boundaries, which can make a substantial difference  in practical applications \cp{fan:96}.  However, in our application, these estimators possess some undesirable properties when combined into matrix estimators $\hSig^m(x)$, $m = 0,1$.  First, while both lead to symmetric matrices, neither is guaranteed to be nonnegative definite, so that standard analyses involving covariance matrices cannot be applied without some post-hoc alterations.  

In the case of the Nadaraya-Watson estimator, this can in fact be resolved by setting a common bandwidth $h_{jk} = h$. Since the weights $w_{in,0}$ are strictly positive and the space of covariance matrices is convex, this will give a coherent unified estimator $\hSig^m(x) \in \mcSp$.  However, this apparent fix does nothing to resolve poor performance near the boundaries.  On the other hand, even with a common bandwidth, the local linear estimator can still fail to be nonnegative definite, since $w_{in,1}$ can be (and often are) negative, especially near the boundaries.  If one requires a true covariance, a natural procedure would be to project onto the space of covariance matrices by truncating negative eigenvalues to zero.  In the next section, we justify this approach by reframing the covariance estimation problem as a geometric problem involving Fr\'echet means and an adapted version of a recently developed regression method for responses that are random objects in a metric space.

\subsection{Time-varying Covariance Matrices as Conditional Fr\'echet Means}
\label{ss: cfm}

A key step is to characterize the matrix $\Sigma(x)$ as a conditional Fr\'echet mean.  The Fr\'echet mean of a random element $Z$ of an arbitrary metric space $(\Om, d)$ is
\begin{equation}
\label{eq: fr_mean}
E_\oplus(Z) = \argmin_{\om \in \Om} E\left(d^2(Z, \om)\right),
\end{equation}
where existence can be guaranteed by compactness, but uniqueness cannot.  If $\Om$ is a convex subset of a Euclidean space, and $d$ is the Euclidean metric, then the ordinary mean and Fr\'echet mean are known to coincide. Let $U$ be a random $p$-vector with zero mean and covariance matrix $\Lambda$ and set $Z = UU^\top.$  Then, if $\Om = \mcSp$ and $d = d_F$ is the Frobenius distance, it can be easily shown that the unique Fr\'echet mean of $Z$ is $E_\oplus(Z) = \Lambda$ due to the fact that $d_F$ is a Euclidean norm as it arises from an inner product.  

Now, let $X$ and $U$ be generic copies of the sample elements $X_i$ and the latent multivariate processes $U_i$,  $Y = U(X)$ and $C(X) = (Y - \mu(X))(Y - \mu(X))^\top.$  Then, for the Frobenius metric $d_F,$ the matrix $\Sig(x)$ in (\ref{eq: mean_cov}) can naturally be expressed as
\begin{equation}
\label{eq: Sig_frechet}
\Sig(x) = E\left(C(X)|X = x\right) = E_\oplus\left(C(X)|X = x\right).
\end{equation}

This characterization suggests that, while smoothing is an appropriate method for estimation of $\Sig(x)$, it should not be done on the scalar raw covariances in (\ref{eq: raw_cov}), as this merely targets the individual entries of $\Sig(x)$. Rather, one should smooth entire covariance matrices and target the Fr\'echet mean $\Sig(x).$  As our data consist of multivariate observations, the data objects closest to our targets are the raw covariance matrices
\begin{equation}
\label{eq: raw_cov_mat}
\hat{C}_{i} = (Y_{i} - \hmu(X_{i}))(Y_{i} - \hmu(X_{i}))^\top.
\end{equation}
Before we develop the estimator, we provide two remarks that highlight the unique difficulties and challenges associated with this estimation problem.

\begin{Remark}
\label{rm: bias}
The covariances $\hat{C}_{i}$ are \emph{biased} in the sense that $$E_\oplus(\hat{C}_{i} | X_{i} = x) \neq \Sigma(x).$$  This is because they serve only as approximations for the unobservable quantities 
\begin{equation}
\label{eq: C_known}
C_{i} = (Y_{i} - \mu(X_{i}))(Y_{i} - \mu(X_{i}))^\top,
\end{equation}
for which $E_\oplus(C_{i} | X_{i} = x) = \Sigma(x).$ 
\end{Remark}
\begin{Remark}
\label{rm: degenerate}
Another feature of the raw covariances $\hat{C}_{i}$, which we would still face if we actually observed the unknown $C_{i}$ in (\ref{eq: C_known}), is that they are degenerate, i.e., they have rank one.  The challenge is then that we are attempting to estimate the object $\Sigma(x)$ that lies in the interior of $\mcSp$ using a sample of objects  that, by definition, reside on the boundary of the space.  
\end{Remark}

In recent years, the smoothing of covariance matrices has been studied in various contexts, with two prominent examples being the mapping of neural pathways using DTI \cp{yuan:12,carm:13} and the spatial modeling of speech recordings \cp{tava:16}.  In  the DTI application, the data consist of a random sample of pairs $(P_i, S_i)$, where $P_i \in \mc{R}$ and $S_i$ is a symmetric positive definite matrix, and the target is the conditional Fr\'echet mean $E_\oplus(S_i | P_i = p)$ under two particular Riemannian metrics for $\mcSp.$  Being based on the sampling of non-degenerate covariance matrices, due to the availability of much richer data in this context, these previous approaches were not affected by  the challenges summarized in  Remarks~\ref{rm: bias} and \ref{rm: degenerate}. Furthermore, since Fr\'echet means generally depend on the chosen metric, in these applications specific features of the chosen metric could be exploited.   This is in contrast to our situation, where the goal is to gain an understanding of the  time-varying covariability of multivariate data. The canonical metric for this purpose is the Frobenius or $L^2$ metric for covariance matrices.

The importance of estimating time-varying covariance matrices is also evident in the spatial analysis of sound objects in \ci{tava:16}, where the  covariance also has a spatial component.  Here, the smoothing of covariances takes place only over space and, at least conceptually, the objects that are smoothed are not random but fixed.  In practice, these objects need to be estimated, and assumptions are included under which the estimation error is negligible.  Such assumptions are common in functional data settings, but do not apply to the cross-sectional data that we consider here, where the raw data are decidedly inconsistent.   In short, the nature of our data requires overcoming the challenges that are highlighted in the above remarks, whereas these challenges do not arise in covariance estimation settings that have been previously studied. 

\ci{tava:16} also introduced the so-called $d$-covariance for multivariate data.  In our setting, if $d$ is a metric on $\mcSp$, the $d$-covariance is
\begin{equation}
\label{eq: dcov}
\Sigma_\oplus^d(x) = \argmin_{S \in \mcSp} E\left(d^2(C(X), S)|X = x\right).
\end{equation}
The motivation for the $d$-covariance is that other metrics besides the Frobenius metric are more suitable for analysis on $\mcSp$, which is a nonlinear  Riemannian manifold.  As \ci{tava:16} demonstrated,  if $d$ is the square-root metric \mbox{$d(C_1, C_2) = d_F(C_1^{1/2}, C_2^{1/2})$,} one can readily use linear methods on the square-root space, which can be identified with the linear space of symmetric matrices.  The $d$-covariance thus provides a natural class of covariance objects that can be estimated using Fr\'echet methods.  Alternative metrics that would have similar limitations as the $d$-covariance are those in the Box-Cox class, where
\mbox{$d_{\alpha}(C_1, C_2) = d_F(C_1^{\alpha}, C_2^{\alpha})$}, with appropriately defined matrix logarithms for the case $\alpha =0$ \cp{pigo:14,mull:16:2}.

However, since $\Sigma(x) = \Sigma_\oplus^d(x)$ if and only if $d$ is the Frobenius metric, one cannot smooth under an alternative metric without losing the ability to interpret the target as a true covariance of the observed random vector.  In the particular application of these methods to the study of myelination that is our primary motivation,  we work with the Frobenius metric for two reasons.  First, practitioners are more comfortable with assessing ordinary covariance and correlation between myelin levels due to its well-established interpretation.  Second, many statistical models that one might wish to apply to the study of myelin levels necessitate an estimate of their covariance properties, which would not be provided  when applying a smoothing procedure under a different metric, as this results in estimates that target a  different quantity.  The motivating problem of regressing ELC scores, which assess cognitive development, on myelin levels as predictors provides one such instance.  

While 
we consider Fr\'echet methods under the Euclidean or Frobenius metric, it bears emphasizing that our approach is not restricted to this choice.  Indeed, the theory in Section~\ref{sec: theory} is based on results in  \ci{mull:18:3}, which hold under very general conditions and can readily be extended to the Fr\'echet estimation of $d$-covariances for various metrics $d$.

\subsection{Fr\'echet Estimation}
\label{ss: frechet_est}

The expression of $\Sigma(x)$ as a conditional Fr\'echet mean in (\ref{eq: Sig_frechet}) suggests estimation by a locally weighted sample Fr\'echet mean \cite[see, e.g.,][]{pigo:14}  
\begin{equation}
\label{eq: est_NW}
\SNW(x) = \argmin_{C \in \mcSp} \sn K_h(X_{i} - x)d_F^2(C, \hat{C}_{i}).
\end{equation}
The individual components of this estimator coincide with those in (\ref{eq: kern_smooth}) for $m = 0$, with weights as in (\ref{w0}) 
and  all bandwidths $h_{jk}$ being equal.  So, while the standard Nadaraya-Watson smoother does not incorporate the constraints of the space $\mcSp$, it results in a standard Fr\'echet estimate of $\Sigma(x)$,  due to convexity.  As mentioned before, this estimator does not work well near boundaries due to bias, and a local linear type estimator based on weights (\ref{eq: weights}) is preferable.

Local linear kernel estimation does not immediately generalize to nonlinear spaces, regardless of whether they are convex or not, as it introduces negative weights near the boundaries, which is necessary to control bias.  In \ci{yuan:12}, intrinsic local polynomial regression (ILPR) was proposed using manifold features of $\mcSp$ under metrics different from $d_F.$  Since the Frobenius metric is the natural choice given our target, ILPR is not directly applicable and, more importantly, more complex than necessary.  A simpler, yet more general, re-characterization of local linear smoothing for arbitrary metric spaces (Local Fr\'echet Regression or Fr\'echet smoothing) has recently been developed in \ci{mull:18:3}.  Fr\'echet smoothing can be easily  applied to the current problem of smoothing covariance matrices, where it becomes 
\begin{equation}
\label{eq: est_LF}
\SLF(x) = \argmin_{C \in \mcSp}\sn w_{in,1}(x, h)d_F^2(C, \hat{C}_{i}),
\end{equation}
and is represented as a weighted Fr\'echet mean, with weights $w_{in,1}$ as in 
(\ref{eq: weights}).

In Fr\'echet smoothing, we are not weighting the covariances directly, but rather the squared distances which define the Fr\'echet functional to be minimized.  Again,  these weights can be negative, since they are derived from the local linear weights in a standard Euclidean smoothing problem, so that $\sn w_{in,1}\hat{C}_{i}$ is not guaranteed to be in $\mcSp$.  However, similar to the Nadaraya-Watson estimator, the local Fr\'echet estimator has an analytic expression in the case of covariance matrices and the Frobenius metric, given by
$$\SLF(x) = \Pi_{\mcSp}\left(\sn w_{in,1}(x, h)\hat{C}_{i}\right),$$ where $\Pi_{\mcSp}$ is the projection operator onto $\mcSp$ under $d_F$.  This projection operator is easy to compute, as it corresponds to truncating negative eigenvalues to zero while keeping the eigenspaces the same. It is interesting, albeit unsurprising, that, while $\SLF(x)$ as defined in \eqref{eq: est_LF} is clearly an intrinsic estimator as a minimizer of a Fr\'echet functional, it is also equivalent to the projection onto $\mcSp$ of the ordinary, extrinsic local linear estimator.  This phenomenon is explained by the fact that Fr\'echet means under Euclidean metrics coincide with ordinary means, along with the uniqueness of Euclidean projections onto closed, convex spaces such as $\mcSp.$

\section{Regional Co-evolution of Brain Myelination in the Developing Brain} 
\label{sec: bambam}

\subsection{Background and the BAMBAM Study}

Altered early brain development is hypothesized to underlie many neurological, behavioral, and intellectual disorders.  Infancy and early childhood are sensitive periods of brain growth, coinciding with the emergence of nearly all cognitive, behavioral, and social functioning abilities.  Beginning in utero, myelination advances rapidly over the first 2 years of life in a carefully choreographed caudal-cranial, posterior-to-anterior pattern \cp{d06,d07}, and continues throughout childhood and adolescence \cp{d08}.  This maturation pattern is driven by the tight regulation of myelination by neural activity \cp{d09,d010,d011}, which likely underpins its spatio-temporal coincidence with emerging neurobehavior \cp{d012,d013,d43}.

The prolonged nature of myelination imparts a high degree of flexibility and plasticity to developing neural systems.  However, this protracted timeline of development also places these systems at prolonged risk to injury or deviant development \cp{d014}.  It is increasingly recognized that alterations in myelination timing and/or extent can significantly affect behavioral and cognitive outcomes \cp{d012}, with altered white matter microstructure being a consistent finding in many neurological and neuropsychiatric disorders \cp{d12, d015,d016}. Despite the importance of coordinated neural communication, relatively little is known regarding the development of structural and functional networks in infants or toddlers or  the relationships linking the emergence of neural networks to evolving cognition or neurobehavioral outcomes.

The Brown University Assessment of Myelination and Behavior Across Maturation (BAMBAM) study provides voxelwise data 
on the myelination for neurotypical children that were very sparsely measured across age, and which we analyze in Section~\ref{ss: bambam} to demonstrate some of these issues. Only children with measurements before the age of 1250 days were included, resulting in a sample size of  $n=223$ children.  Additionally, for those children who were assessed at multiple time points during this period, one observation per subject was selected randomly. 

Imaging measures include anatomical $T_{1}$-weighted ($T_{1}w$), quantitative relaxometry ($qT_{1}$, $qT_{2}$),  and quantitative myelin water fraction (MWF), among other measurements. The children included in the study were  devoid of major risk factors for neurologic and psychiatric disorders.  Children with in utero alcohol or illicit substance exposure, premature ($<37$ weeks gestation) or multiple birth, fetal ultrasound abnormalities, complicated pregnancy, APGAR scores $< 8$, NICU admission, neurological disorder, psychiatric or developmental disorders in the infant, parents or siblings were excluded \cp{d030,d031}. 

\begin{figure}
\includegraphics[scale=0.57]{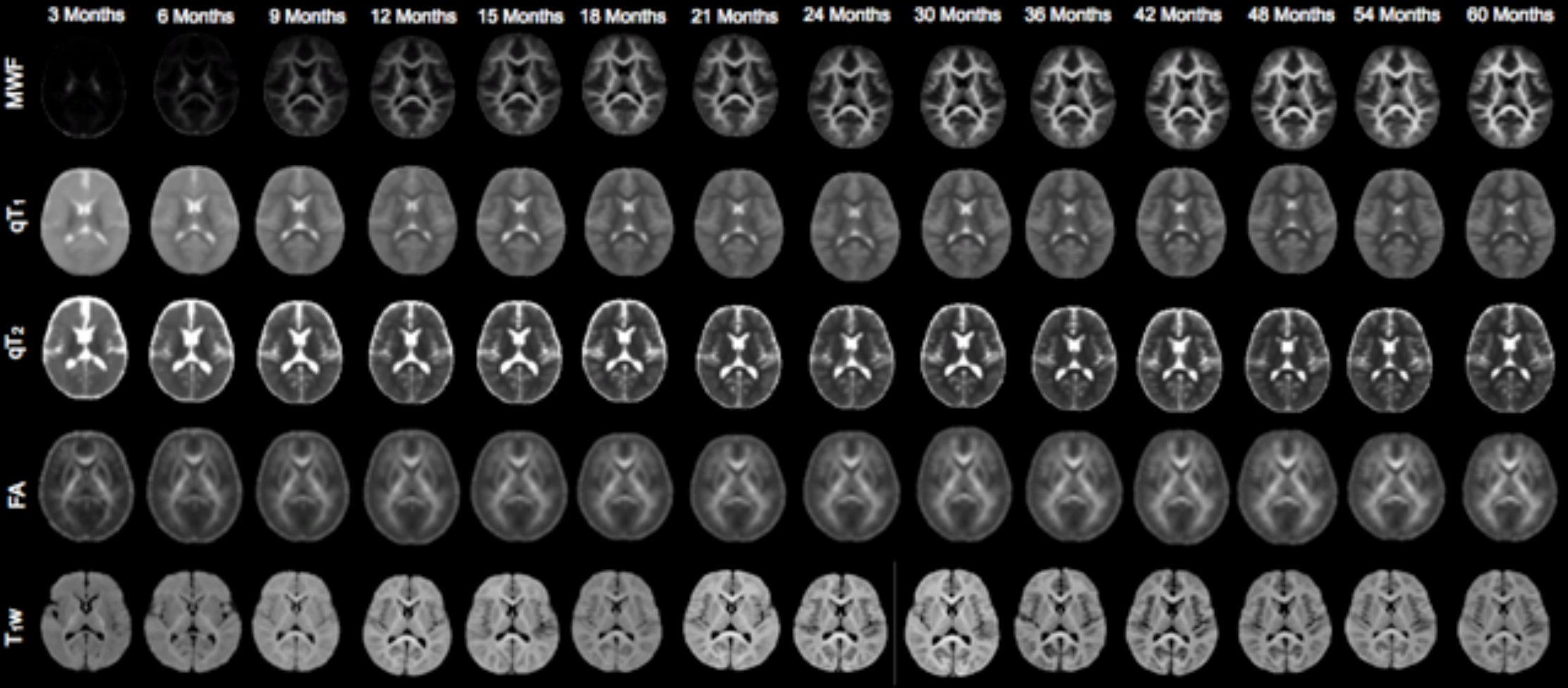} 
\caption{\small{A multi-modal assessment of development.  Mean myelin water fraction (MWF), quantitative relaxometry ($\textrm{qT}_1,$ $\textrm{qT}_2$) and fractional anisotropy (FA) maps, and $\textrm{T}_1$-weighted ($\textrm{T}_1$w) images from 3 months to 5 years of age.}} 
\label{mri}
\end{figure}

The BAMBAM study acquired high-quality and artifact-free MRI data that inform on different characteristics of the developing brain \cp{d45,d47,d62,d63,d68}, including morphometric (whole-brain, white and gray matter volume), myelin water fraction (which informs more specifically on myelination, a key process of brain connectivity), and other modalities \cp{d113,d116,d143}.  Such high quality MRI measurements (see Figure \ref{mri}) are a prerequisite to assess brain development \cp{d30,d161,d165,d180,d181,d182,d183,d186,d187}. It is also of interest to relate myelin level correlation across brain regions  to concurrent  developmental and cognitive scores \cp{d61}, where we use  the Early Learning Composite (ELC), derived from the  verbal and non-verbal development quotient (VDQ and NVDQ, respectively) scores of the Mullen Scales of Early Learning \cp{d124}; these scores are standardized for age.   It is of special interest to identify the brain regions for which myelination levels  are most predictive of these scores at certain ages. This will help to pinpoint  the relative relevance of brain regions for neurocognitive performance across age, a goal that has been elusive so far.

The implementation of Fr\'echet regression that we propose here is a one-step approach that leads directly to bona fide correlation/covariance matrix functions and is supported by theoretical consistency and rate of convergence results that are applicable even for the extremely sparse data case that we consider here. We demonstrate that also in practice this method leads to sensible results when applied to the very sparse data case, where the observed data do not correspond to a sample covariance matrix paired with each predictor, but rather just one random vector of observed levels for each predictor level. 
Our goal is to quantify the evolution of the cross-correlations with age from data that contain only one measurement per subject. For this, we select $p=21$ discrete anatomical brain regions as an example. The methods generalize to other local aggregations of the MWF levels, which are recorded at the voxel level.  These regions are listed in Table~\ref{tab: regions} together with their acronyms used throughout the remainder of the paper.

While modeling cross-correlation and cross-covariance of myelination has not been well explored, 
a related topic with more substantial work is the modeling of the myelin trajectories as a function of age. Nonlinear approaches have been applied previously to model overall myelination level trajectories  by averaging over brain regions and focusing on the dependence of myelin levels as a function of age \cp{d28,d30,d43, d142,d143}. Recently, \ci{mull:17:9} proposed a method to obtain whole brain MWF trajectories, modifying approaches of \ci{mull:10:6,sent:11}. For the case where repeated measurements are available for each subject, ideally densely spaced, related work on modeling multivariate trajectories 
by functional data analysis methods includes \ci{zhou:08,mull:16:13}. 
However, these approaches are not applicable if only 
one or very few observations per subject are available.

\begin{table}[t!]
\scriptsize
  \centering
  \caption{List of all $p = 21$ regions/tracts and their acronyms for which MWF levels are measured. \label{tab: regions}}
  \begin{tabular}{|l c|l c|}
  \hline
    Region & Acronym & Region & Acronym \\ \hline
    body Corpus Callosum &  bCC & left Cingulum & lC \\
    genu Corpus Callosum & gCC & right Cingulum & rC \\
    splenu Coropus Collosum & sCC & left Corona Radiata & lCR \\
    left Frontal & lF & right Corona Radiata & rCR \\
    right Frontal & rF & left Internal Capsule & lIC \\
    left Parietal & lP & right Internal Capsule & rIC \\
    right Parietal & rP & left Optic Radiation & lOR \\
    left Occipital & lO & right Optic Radiation & rOR \\
    right Occipital & rO & left Superior Longitudinal Fasciculus & lSLF \\
    left Temporal & lT & right Superior Longitudinal Fasciculus & rSLF \\
    right Temporal & rT & & \\
    \hline
  \end{tabular}
\end{table}

\subsection{Age-varying correlations of regional myelination}
\label{ss: bambam}

The BAMBAM study cohort includes $n = 233$ children who were observed during the first $1250$ days of life, with the earliest observation at age 65 days.  Some children were observed repeatedly within this time period for a total of $N = 440$ visits, although nearly half (48\%) of the children have just one measurement and 90\% have at most 3 measurements.  For this reason, in our analysis we assume that only a single observation per subject is available, which is chosen uniformly from the available observations.  For comparison purposes, we also conducted this  analysis using the full data set, with corresponding figures being relegated to the Appendix.

We investigated the patterns of MWF development in the $p = 21$ regions/tracts listed in Table~\ref{tab: regions}.  A subset of the observed data are shown in the left panel of Figure~\ref{fig: MWF_traj}, where the children were binned into thirty-five age groups of equal width, and one child per bin was chosen randomly for display.  The right panel of Figure~\ref{fig: MWF_traj} also shows the preliminary estimates of the mean MWF levels, using local linear smoothing with a bandwidth of $50$ days.  This bandwidth was chosen by five-fold cross validation combined across all regions. Comparison with Figure 7 
in \ci{pete:18} shows similar patterns for estimates obtained from the complete data set.

\begin{figure}[t]
\centering
\subcaptionbox{}[2.05in]{\includegraphics{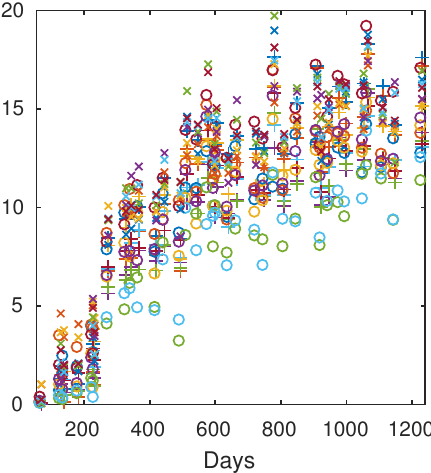}} \subcaptionbox{}[2.85in]{\includegraphics{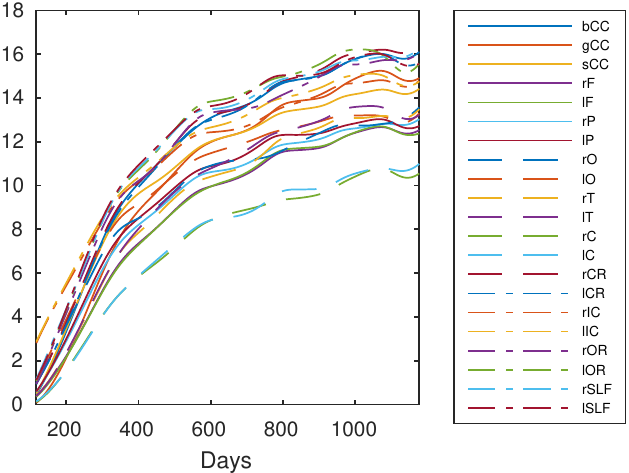}}
\caption{\small (a) Scatterplot of MWF levels for a subset of $20$ children, with different plotting characters and colors corresponding to distinct brain regions. (b) Estimated mean functions for all 21 regions. \label{fig: MWF_traj}}
\end{figure}

\begin{figure}[h!]
\centering
\subcaptionbox{$x = 250$ days}[2.45in]{\includegraphics{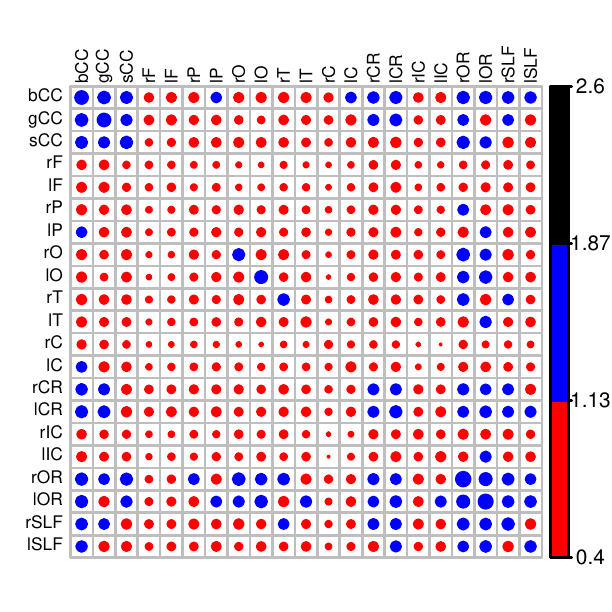}} \subcaptionbox{$x = 500$ days}[2.45in]{\includegraphics{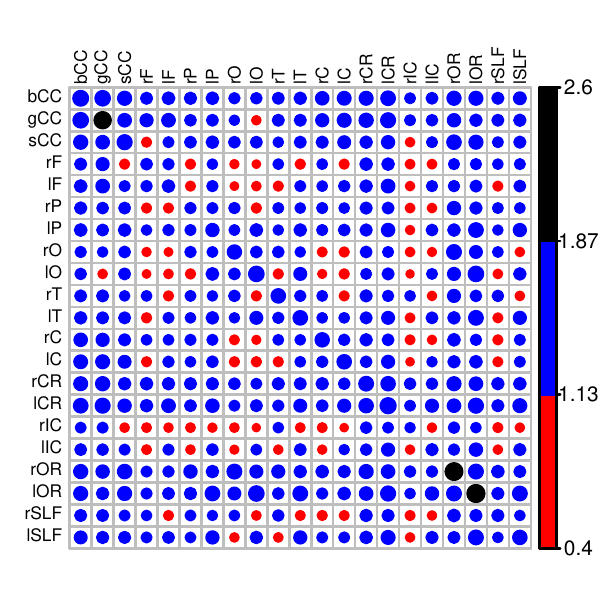}} \\
\subcaptionbox{$x = 750$ days}[2.45in]{\includegraphics{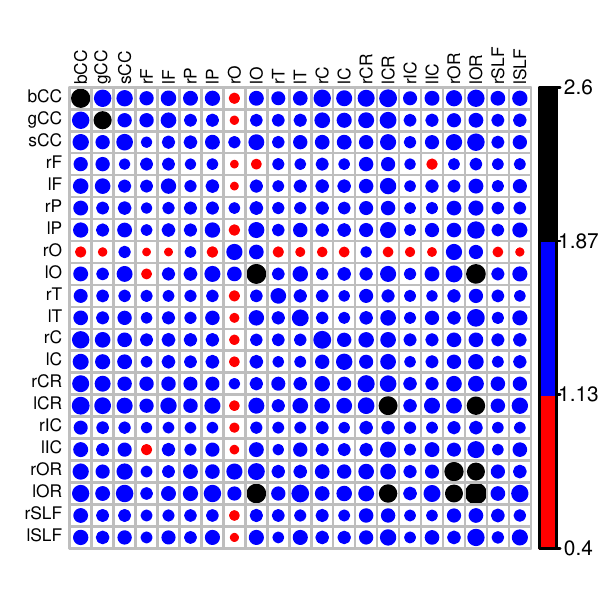}} \subcaptionbox{$x = 1000$ days}[2.45in]{\includegraphics{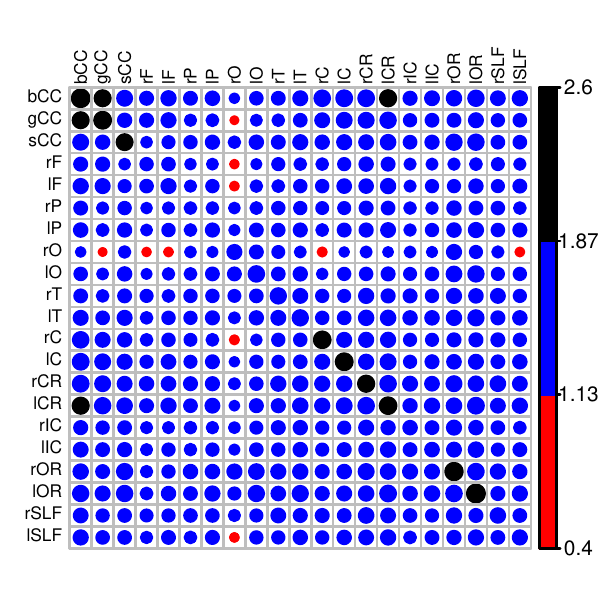}}
\caption{\small Estimated MWF covariance matrices at four distinct ages $x$ (in days). Circle size indicates magnitude of covariance (all are positive), and the color/shade indicates a discretization of the covariance values for easier visual comparison. \label{fig: cov_cvs}}
\end{figure}

We then computed cross-sectional covariance estimates as in \eqref{eq: est_LF}; {the alternative estimators in \eqref{eq: kern_smooth} and \eqref{eq: est_NW} are not shown due to their inferior performance, which is expected theoretically and empirically demonstrated in the simulations of Section~\ref{sec: sims}.  For bandwidth choice, a standard cross validation criterion would not be expected to perform well as the raw covariance matrices are known to live on the boundary of the space (see Remark~\ref{rm: degenerate}). Letting $\hat{\Sigma}_{-i}(X_i)(h)$ denote a leave-out estimate of $\Sigma(X_i)$ using bandwidth $h$, we consider the following criteria, 
\begin{align}
\label{eq: bw_choice}
\hat{h}_1 & = \argmin_h \sum_{i = 1}^n d_F^2(\hat{C}_i,\hat{\Sigma}_{-i}(X_i)(h))^2, \nonumber \\
\hat{h}_2 &= \argmin_h \sum_{i = 1}^n \textrm{tr}(\hat{\Sigma}^{+}_{-i}(X_i)(h)\hat{C}_i),
\end{align}
where $A^+$ denotes the pseudo inverse of a matrix $A$.  As expected, $\hat{h}_1$ tends to undersmooth the data, whereas $\hat{h}_2$ oversmooths due to the inversion.  We found that using the geometric mean $\hat{h}_\textrm{opt} = (\hat{h}_1\hat{h}_2)^{1/2}$ provides a reasonable choice of 150 days for the covariance smoothing, using five-fold cross validation. 

\begin{figure}
\centering
\subcaptionbox{\label{fig: cor_box}}[2.45in]{\includegraphics{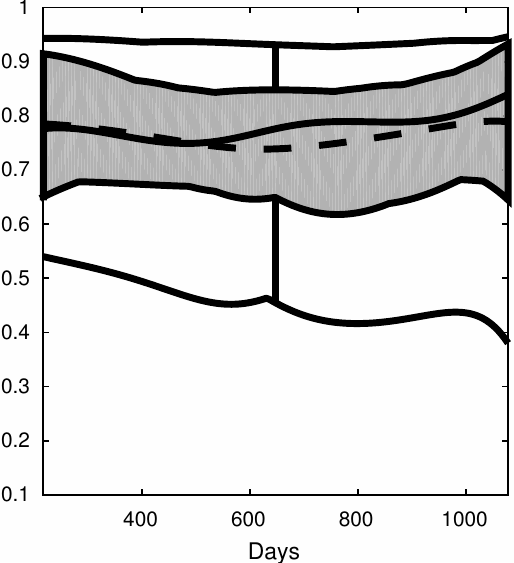}} \subcaptionbox{\label{fig: cor_eig}}[2.45in]{\includegraphics{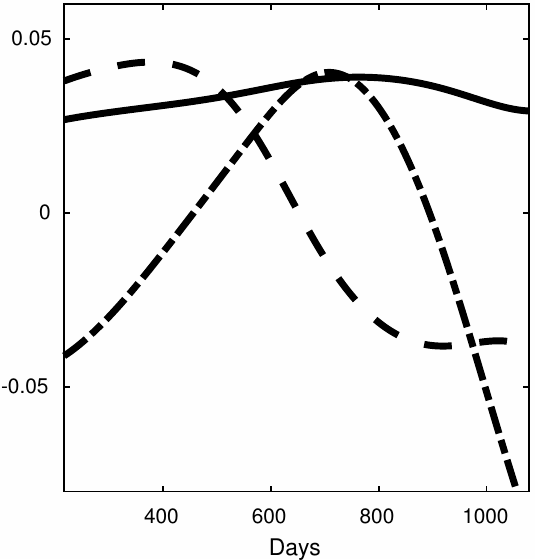}} \\
\subcaptionbox{\label{fig: cor_scores}}[2.45in]{\includegraphics{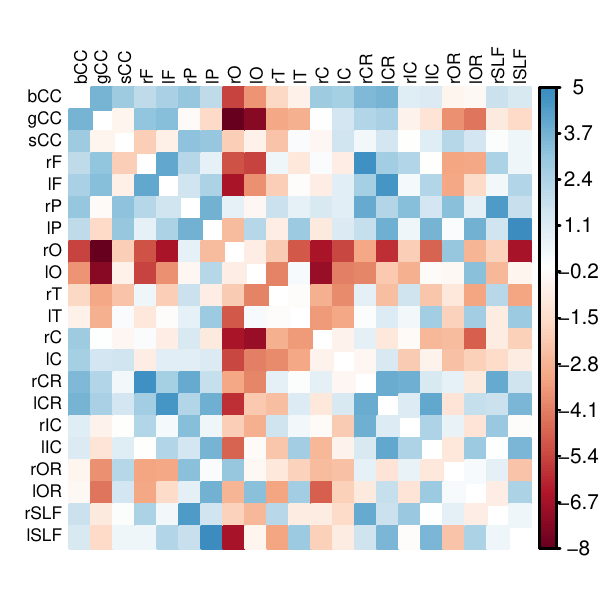}}
\caption{\small (a) Functional boxplot and mean curve (dashed), for correlation curves between the 21 regions. (b) First (solid, FVE $ = 90.1\%$), second (dashed, FVE $ = 6.15\%$) and third (dot-dash, FVE $ = 2.76\%$) eigenfunction representing the variability in the collection of estimated correlation curves. (c) First principal scores for all region pairs. Here FVE stands for Fraction of Variance Explained \cite[for Functional Principal Component Analysis and related issues, see, e.g.,][]{mull:16:3}.
 \label{fig: cor_fpca}}
\end{figure}

We visualize the $231$ distinct covariance curves $\hat{\Sigma}_{jk}(x)$, $1 \leq j \leq k \leq 21$ for four ages  in Figure~\ref{fig: cov_cvs}.  We see from the diagonals that variability of MWF levels increases with age, while covariability also scales but to a lesser extent.  An interesting finding is the weak dependence of the right occipital region with all other regions, despite having a similar variability pattern.   We can also examine the dependency patterns while normalizing the scale by considering the $210$ distinct \emph{correlation} curves $\hat{\mc{R}}_{jk}$, $1 \leq j < k \leq 21$, this time using Functional Principal Component Analysis (FPCA).  We note here that the correlation curves are dependent as they are constructed from the same individuals so that many of the basic results such as consistency of FPCA are not applicable. We also ignore the constraint that correlation curves are between -1 and 1 when implementing FPCA. However, none of this precludes us to use FPCA for purely descriptive purposes, and this is the premise for the following FPCA application. 

 Figure~\ref{fig: cor_box} shows  the pointwise average of these correlation curves along with  a functional boxplot. The mean curve indicates that correlation between regional myelin development remains roughly constant, when averaging across all region pairs.  Comparing these results with the estimates for the full data set, we again find consistency in the dynamics of covariance estimates (Figure~\ref{fig: cov_cvs} and Figure 8 in \ci{pete:18}),
as well as similar patterns in the corresponding correlation curves (Figure~\ref{fig: cor_fpca} and Figure 9 in \ci{pete:18}).

The first $k = 3$ eigenfunctions representing the principal deviations from  the average correlation curve observed in the estimated correlations are shown in Figure~\ref{fig: cor_eig}; the fractions of variance explained (FVE) by these three curves are $90.1\%$, $6.15\%$ and $2.76\%$, respectively.  The first eigenfunction has an increasing trend up to day 800, followed by a decreasing trend. This indicates that variability between correlation patterns amongst distinct region pairs is mostly explained by an increased separation from the mean correlation up until approximately 800 days of age, followed by a regression to the mean, i.e., moving towards the average correlation.  The eigenfunction curves in Figure~\ref{fig: cor_eig} can be further interpreted by examining the corresponding functional principal component scores in Figure~\ref{fig: cor_scores}.  As previously observed, the occipital regions have markedly lower correlations with the remaining regions as indicated by the negative scores in this component. It can also be seen that the strongest correlations in region pairs tend to include the Frontal and Parietal regions as well as the Corona Radiata.

We now aim at modeling the relation between myelin development and cognitive ability, as quantified by the ELC score, building on the above approaches. Denoting 
by  $E_i(x)$ a latent process that tracks the ELC score of the $i$th child continuously over the age domain, we are able to actually observe only $E_i(X_i),$  the cognitive score at the random time $X_i$ at which the child is studied. Recall that $U_i(x)$ is the vector latent process containing the MWF levels for the various regions across age, with mean curve $\mu(x).$  Since both cognition and myelin development are  dynamic processes, we consider a varying coefficient linear model \begin{equation}
\label{eq: ELC_model}
E_i(x) = \beta_0 + \beta^\top(x)(U_i(x) - \mu(x)) + \epsilon_i,
\end{equation}
where in our specific application we can substitute fixed intercept value $\beta_0=100$, owing to the fact that ELC scores are standardized to have a mean value of 100 and constant variance.  

In view of  the linear structure of (\ref{eq: ELC_model}), 
the solution for the slope vector function is
$$
\beta(x) = \Sigma\inv(x)\Gamma(x),
$$
where $\Gamma(x) = \Cov(U_i(x),E_i(x))$ is the dynamic vector of covariances between myelin development and the ELC score., and $\Sigma(x)$ is the time-varying covariance matrix that we are targeting with the approach described above. 

We then propose to estimate the parameter functions $\beta$ in \eqref{eq: ELC_model} by first estimating the elements $\Gamma_j(x)$, $j =1 ,\ldots,21.$ For this, we compute  $G_{ij} = (Y_{ij} - \hat{\mu}_j(X_i))(E_i(X_i) - 100)$ and then apply local linear smoothing for pairs $(X_i, G_{ij})$ with bandwidth $h = 150$ days.  Finally, we use the estimate $\hat{\Sigma}(x)$ obtained by Fr\'echet smoothing to compute the regularized ridge estimate
\begin{equation}
\label{eq: ELC_est}
\hat{\beta}_\lambda(x) = \left[\hat{\Sigma}(x) + \lambda I_p\right]\inv \hat{\Gamma}(x).
\end{equation}
The bandwidths for the estimation of $\mu$ and $\Sigma$ were set to $50$ and $150$, as done previously, and the ridge parameter was chosen to minimize prediction error using five-fold cross validation, resulting in $\hat{\lambda}_{\textrm{opt}} = 0.14$.  Estimates $\hat{\Gamma}(x)$ and $\hat{\beta}_{\hat{\lambda}_{\textrm{opt}}}$ are shown in Figures~\ref{fig: gamma_est} and \ref{fig: beta_est}.  While the covariability between each MWF level and ELC score has similar patterns over time, the interdependencies betwen the MWF levels cause the estimates of the components of the vector function $\beta$ to be considerably more complex. 

For example, these estimates indicate a stark contrast in the effects of myelination level in the Genu of the Corpus Callosum, which is associated with a lower cognitive score, versus white matter in the right Parietal lobe, which is associated with a higher score, particularly pronounced at 600 days, and also between the right Optic Radiation, with lower scores associated, and left Corona Radiata, with higher scores associated and most pronounced at 800 days. Despite the aforementioned similarities in mean and covariance estimates between the reduced and full data sets, comparison of the functional coefficients $\beta$ shown in Figure~\ref{fig: ELC_est} above and Figure 10 in \ci{pete:18} demonstrate their sensitivity to changes in $\hat{\Sigma}(x)$ due to the inversion in \eqref{eq: ELC_est}.

\begin{figure}[ht!]
\centering
\subcaptionbox{$\hat{\Gamma}$ \label{fig: gamma_est}}[2.95in]{\includegraphics{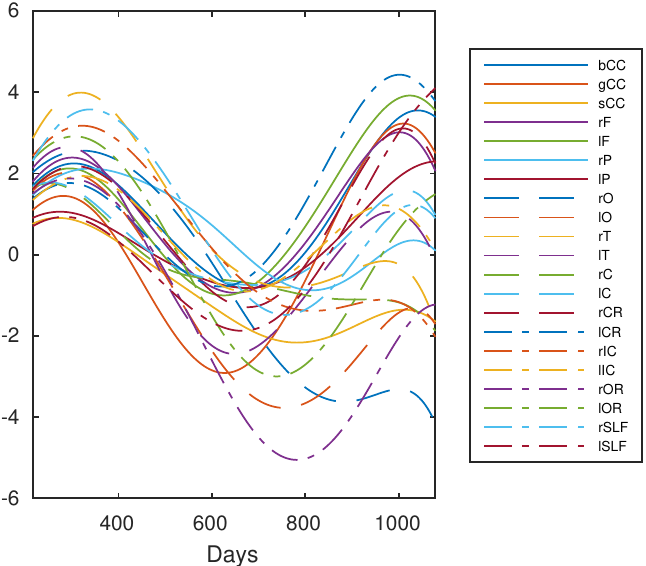}} \subcaptionbox{$\hat{\beta}_{\hat{\lambda}_{\textrm{opt}}}$\label{fig: beta_est}}[1.95in]{\includegraphics{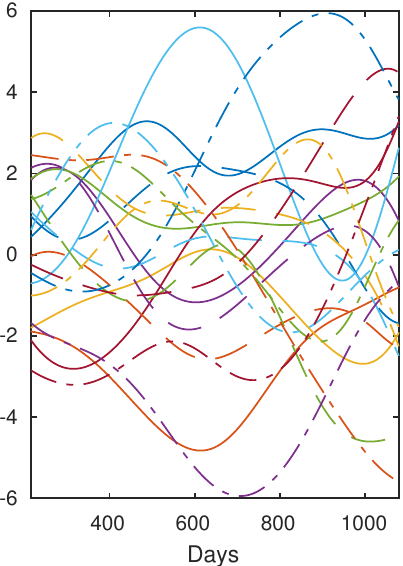}} \\ 
\subcaptionbox{$\hat{R}^2$\label{fig: rsq_est}}[2.45in]{\includegraphics{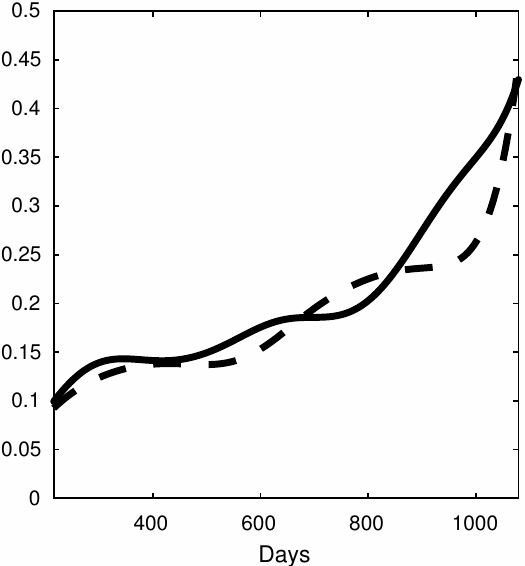}}
\caption{\small Estimates $\hat{\Gamma}(x)$ in (a), $\hat{\beta}_{\hat{\lambda}_{\textrm{opt}}}(x)$ in (b) and $\hat{R}^2(x)$ in (c), corresponding to the varying coefficient  model \eqref{eq: ELC_model} for regressing ELC scores on myelin water fraction levels. In panel (c), estimates based on both one observation per subject (solid line) and the full data set (dashed line) are shown. \label{fig: ELC_est}}
\end{figure}

Another meaningful measure we can extract is an estimate of the time-varying coefficient of determination, or multiple correlation,
$$
R^2(x) = \Gamma(x)^\top \Sigma(x)\inv \Gamma(x),
$$
which we estimate by 
$$
\hat{R}^2(x) = \hat{\Gamma(x)}^\top \left[\hat{\Sigma}(x) + \hat{\lambda}_{\textrm{opt}}I_p\right]\inv \hat{\Gamma(x)},
$$
with the resulting estimates shown  in Figure~\ref{fig: rsq_est}, where also the  corresponding estimates using the complete data set are shown.  One finds that while  MWF levels only moderately account for variability in ELC score at younger ages, there is a sharp and steady increase in their predictive power between the ages of two and three.  

These exploratory results clearly demonstrate the evolving relationships between maturing brain structure and developing cognitive functioning and behaviour, particularly throughout early neurodevelopment.  In older children, adolescents, and adults, neuroimaging results generally conform to the hypothesis that increased brain structure, integrity, and  myelin are associated with improved cognitive function.  Our results suggest that this relationship is less clear and exhibits complex nonlinearities  during infancy and early childhood.  For example, our finding that reduced myelination within portions of the Corpus Callosum is associated with increased cognitive functioning runs contrary to preconceptions derived from adult studies \cp{frye:08}, however, may be a natural consequence of other prior findings suggesting that slower maturation is also associated with improved cognitive outcomes \cp{d142,erus:14}.

With respect to our finding that brain structure and function relationships become increasingly stronger with age, this too may provide insight into the increasing importance of white matter with age, with grey matter structure playing the more important role in earlier life \cp{smys:16}.  However, it is also important to note that early cognitive measures suffer high variability and poor predictive ability prior to 12-18 months of age \cp{slat:97}, so that this finding may also reflect intrinsic measurement variability.  Clearly, in addition to the estimation techniques introduced here, it will be useful in future work to develop inferential methods for testing and confidence sets in order to assess the strength of evidence for these findings.

\section{Simulations}
\label{sec: sims}

The  simulations reported  in this section aim, first, to demonstrate the advantages of local Fr\'echet smoothing over the Nadaraya-Watson smoother and, second, to establish the reliability of estimates obtained from real data in Section~\ref{ss: bambam}.  The simulations were conducted under the setting $T = 1$, with $X_i$ being independently sampled from a beta distribution with shape parameters 0.5 and 1.8, which roughly approximates the shape of the age distribution in the BAMBAM data.   Then, $Y_i|X_i$, $i = 1,\ldots, n$, were generated according to the model
\begin{equation}
\label{eq: sim_model}
Y_i|X_i = \mu(X_i) + \left[\Sigma(X_i)\right]^{1/2}Z_i.
\end{equation}
In order to emulate the BAMBAM data, for a fixed dimension $p$, the mean vector $\mu(x)$ had components $\mu_j(x) = b_j - 8(x - c_j)^2$, $j = 1,\ldots, p$, where \mbox{$b_j\sim \mathcal{U}(10,16)$} and $c_j\sim\mathcal{U}(1,1.1)$ were drawn once and fixed for all simulations.  These parametric mean functions were chosen based on the estimated mean myelin trajectories (see Fig.~\ref{fig: MWF_traj} in Section~\ref{ss: bambam}).  The covariance $\Sigma(x)$ was formed by generating a $p\times p$ matrix $A$ with independent $\mathcal{N}(0, 0.5)$ random variables in each entry, then computing $S = 0.5(A + A^\top).$  A second $p\times p$ matrix $V$ was generated with elements drawn independently as $\mathcal{U}(0, 0.5)$, from which $\theta = 0.5(V + V^\top)$ was computed.  Finally, with $\textrm{Exp}$ denoting matrix exponentiation  and $\odot$ the Hadamard product, we formed
\begin{equation}
\label{eq: sim_cov}
\Sigma(x) = (1 + 10x + 20x^5)\textrm{Exp}\left[S \odot \sin\left(2\pi\theta(x+0.1)\right)\right].
\end{equation}
The polynomial factor emulates the increasing variability of MWF levels seen in the BAMBAM data, with minor fluctuations induced by the sine component. 
Repeated simulations were run by sampling the random variate $Z_i$ independently of $X_i$ as a standard $p$-dimensional Gaussian random vector.  Mean trajectory estimates $\hat{\mu}_j$ were obtained using local linear regression with the bandwidth chosen by ordinary leave-one-out cross validation, pooled across $j$, and separately for each simulation run.  Covariance estimation was performed using both Nadaraya-Watson \eqref{eq: est_NW}, local Fr\'echet \eqref{eq: est_LF}, and standard local linear estimators, which are obtained as in (\ref{eq: kern_smooth}), with $m = 1$ and weights $w_{in,1}$  as specified in 
(\ref{eq: weights}). For the latter,  the various bandwidths were all taken to be equal, i.e.\ $h_{jk} \equiv h.$

To compare the various estimates, for each fixed bandwidth $h$ over a grid, we computed the integrated squared error (ISE)
$$
\text{ISE}(h)=\int_0^1 d_F^2(\hat{\Sigma}(x)(h), \Sigma(x))\; \d x
$$
for each simulation run, where $\hat{\Sigma}(x)(h)$ stands for a generic estimate of $\Sigma(x)$ using bandwidth $h$.  The results for comparing \eqref{eq: est_NW}, \eqref{eq: est_LF}, and \eqref{eq: kern_smooth} using this metric are in Table~\ref{tab: MISE} for 100 simulations, dimensions  $p = 20, 40$, and sample sizes  $n = 250, 500, 1000.$  The settings we investigated include parameter values that approximate the BAMBAM data, and go beyond by increasing both dimension and sample size.  One finds that average ISE values are uniformly lower for the local linear and Fr\'echet methods, and the very small optimal bandwidths used for the Nadaraya-Watson estimator particularly reflect its well-known inherent bias issues, and increasing sample size leads only to very small declines or no declines in ISE, in contrast to the local linear and Fr\'echet methods, where increasing sample sizes lead to noticeable improvements. While the tabulated \emph{integrated} MSE results demonstrate only slight improvement using local Fr\'echet estimation versus ordinary local linear estimation, in fact the \emph{pointwise} discrepancies $d_F(\hat{\Sigma}(x)(h), \Sigma(x))$ are uniformly lower for the local Fr\'echet method, with larger discrepancies near the boundaries where the ordinary local linear estimate tends to suffer from  influential negative eigenvalues.

\begin{table}[t]
\small
  \centering
  \caption{Logarithms of average ISE for Nadaraya-Watson \eqref{eq: est_NW}, ordinary local linear \eqref{eq: kern_smooth}, and local Fr\'echet \eqref{eq: est_LF} estimators, minimized over a grid of bandwidth values $h$.  The minimizing bandwidth  is given in parentheses. \label{tab: MISE}}
  \begin{tabular}{|l|l|c|c|c|}
  \hline
    Dimension  & Method & $n = 250$ & $n  = 500$ & $n = 1000$ \\ \hline
    \multirow{3}{*}{$p=20$} & Nadaraya-Watson & 11.981(0.08) & 11.958(0.06) & 11.937(0.05) \\
    & Local Linear & 10.799(0.43) & 10.430(0.33) & 10.032(0.29) \\
    & Local Fr\'echet & 10.798(0.43) & 10.430(0.33) & 10.031(0.29) \\ \hline
    \multirow{3}{*}{$p=40$} & Nadaraya-Watson & 14.693(0.09)& 14.673(0.07) & 14.655(0.05) \\
    & Local Linear & 13.846(0.60) & 13.457(0.33) & 13.068(0.29) \\
    & Local Fr\'echet & 13.845(0.60)& 13.456(0.37) & 13.068(0.29)\\ \hline
  \end{tabular}
\end{table}

The simulation results are visualized in Figure~\ref{fig: box}, demonstrating the ISE values (in log scale) over all simulation runs rather than just the mean, where we leave out the ordinary local linear estimator.  The local Fr\'echet approach clearly outperforms the Nadaraya-Watson estimator uniformly over all settings, which, again due to its bias,  is seen not to display the improved performance with increasing sample size $n$ that is found for the local Fr\'echet estimator and is expected for a reasonable estimation approach. 

To further validate our analysis in Section~\ref{sec: bambam}, we implemented  two additional estimates of the dynamic covariance for  the setting $n = 250$ and $p = 20.$ The first alternative estimate is the local Fr\'echet estimator under the square root metric $d(C_1,C_2) = d_F(C_1^{1/2},C_2^{1/2})$,
$$
\hat{\Sigma}^d_\oplus(x) = \left(\sum_{i = 1}^n w_{in,1}(x, h) \hat{C}_i^{1/2}\right)^2.
$$
The appeal of this estimate is that it is an intrinsic estimator that does not require projection onto $\mcSp.$  However, as noted in Section~\ref{ss: mfd_methods}, this estimator no longer targets the ordinary covariance, but rather the $d$-covariance of \ci{tava:16} in \eqref{eq: dcov}.  This inherent bias resulted in a worse average ISE of $11.376$ (on the log scale, and minimized over a grid of bandwidths) when compared to local Fr\'echet estimation under the Frobenius metric in Table~\ref{tab: MISE}, although the performance was better than that of the standard Nadaraya-Watson estimate.

\begin{figure}[h]
\begin{center}
\subcaptionbox{$p = 20$}[2.45in]{\includegraphics{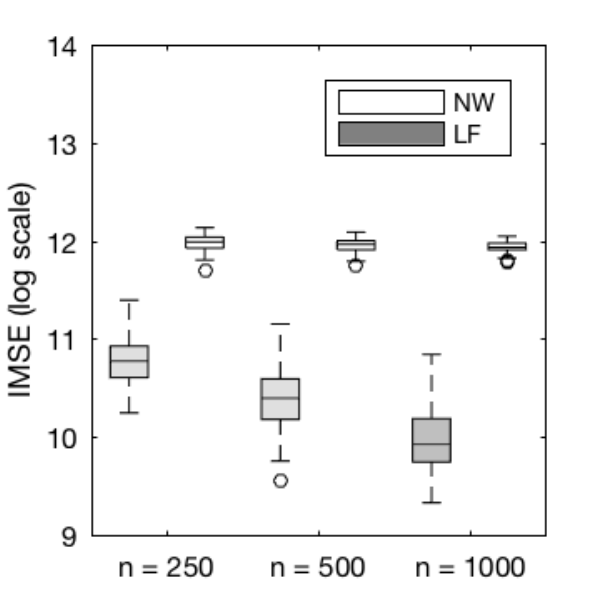}} \subcaptionbox{$p = 40$}[2.45in]{\includegraphics{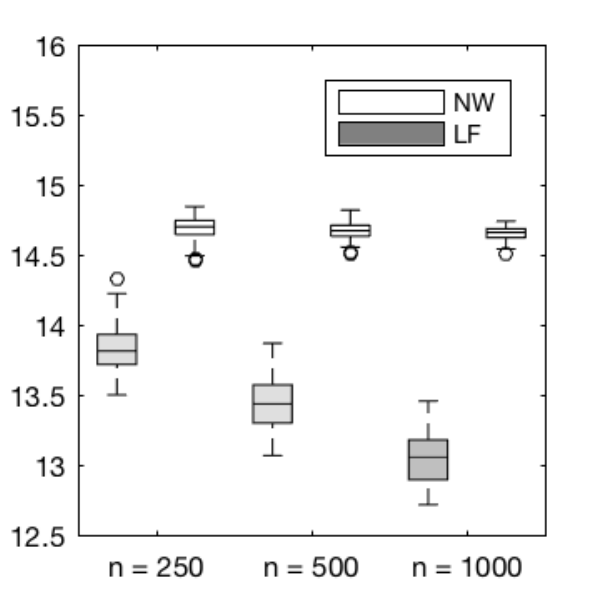}}
\caption{\small Boxplots of integrated squared errors for Nadaraya-Watson (NW) and local Fr\'echet (LF) estimators for $p = 20, 40$ and $n = 250, 500, 1000$, using the optimal bandwidths given in Table~\ref{tab: MISE}. \label{fig: box}}
\end{center}
\end{figure}

A second alternative estimate was considered  in order  to assess the relative efficiency of using just one observation per subject versus the full data set in the BAMBAM analysis. This estimate was obtained by generating repeated observations for the $n = 250$ subjects and then computing the local Fr\'echet estimate using all observations, ignoring dependencies.  To mimic the BAMBAM data, we generated independent obervation counts $N_i \in \{1, 2, 3, 4\}$ with probabilities $(0.48, 0.28, 0.14, 0.1)$, resulting in a total of 497 observations including repeats, and these numbers were used across all simulation runs.  For indices $i$ with $N_i > 1,$ this required an adaptation of the generative model in \eqref{eq: sim_model} in order to include dependencies.  This was done by generating independent timepoints $X_{ij}$, $j = 1,\ldots,N_i,$ according to the beta distribution as before, followed by a zero-mean multivariate normal random vector $(Z_{i1}^\top,\ldots,Z_{iN_i}^\top)^\top$ with $\Var(Z_{ij}) = I_p$ and $\Cov(Z_{ij}, Z_{ij'}) = 0.2I_p,$ $j,j' = 1,\ldots, N_i$, $j' \neq j.$  Lastly, the random vector
$$
Y_{ij}|X_{ij} = \mu(X_{ij}) + \left[\Sigma(X_{ij})\right]^{1/2}Z_{ij}
$$
was computed. Unsurprisingly, the increased information resulted in an improved average ISE of $10.465$ (again in log scale).  Hence, when repeats are available, they can indeed be beneficial.  However, the methodology does not require repeats and, most importantly, the challenges mentioned in Section~\ref{ss: mfd_methods} associated with single observation data are not alleviated for very sparse longitudinal data such as those available in the  BAMBAM study.

\section{Theoretical Justifications}
\label{sec: theory}

We analyze the pointwise behavior of the estimate $\SLF(x)$ and establish the rate of convergence for the metric $d_F(\SLF(x), \Sigma(x)).$  If $\mu$ were known, so that $C_i$ in \eqref{eq: C_known} can be computed, the estimator
\begin{equation}
\label{eq: Sig_known}
\SLFt(x) = \argmin_{C \in \mcSp} \sn w_{in,1}(x, h)d_F^2(C, C_i)
\end{equation}
corresponds to the estimator studied in \ci{mull:18:3}.  Thus, we can make use of the inequality $$d_F(\SLF(x), \Sigma(x)) \leq d_F(\SLF(x), \SLFt(x)) + d_F(\SLFt(x), \Sigma(x)).$$ Suppose that the marginal densities $f_X$ of $X_1$ and $f_U(\cdot, x)$ of $U_1(x)$ exist, the latter for all $x \in [0,T]$.  Let $\lVert \cdot \rVert$ be the standard Euclidean norm on $\mathbb{R}^p,$ and $h = h_n \rightarrow 0$ be a bandwidth sequence.  We require the following assumptions.
\begin{itemize}
\item[(A1)] The kernel function $K$ is a symmetric probability density function with support $[-1,1].$
\item[(A2)] The density $f_X$ of $X_1$  is twice continuously differentiable and \newline $\inf_{x \in [0,T]} f(x) > 0.$
\item[(A3)] There is a constant $M_1 > 0$ such that $\sup_{x \in [0,T]}\lVert U_1(x) \rVert  \leq M_1$ almost surely.  Futhermore, $$\inf_{x\in [0,T], \lVert y \rVert \leq M_1} f_U(y, x) > 0.$$
\item[(A4)] The function $f_U(u, x)$ is continuous with respect to both arguments and twice continuously differentiable with respect to $x$ on the interior.
\item[(A5)]  For any $x \in [0,1]$ and a null  sequence $q_n = o(1),$ the auxiliary estimates $\hat{\mu}_j$ satisfy 
\[
\sup_{|y - x| \leq h} |\hat{\mu}_j(y) - \mu_j(y)| = O_p(q_n).
\]
\end{itemize}

Assumption (A1) is common for estimation of a function with bounded support, but can be relaxed by controlling the tail behavior of $K$.  Assumption (A2) is also widely used in local polynomial estimation settings, while (A3) is a natural assumption which has the implication that the conditional density of $Y_1|X_1$ is well-behaved.  Assumption (A4) can also be relaxed by controlling the pointwise tail behavior of the process $U_1$. The required rate in (A5) is needed to control the error incurred by using $\hat{C}_i$ in place of $C_i$ in the estimation, where the upper bound $q=n^{-2/5}\log(n)$ is known to hold for local linear estimators under certain assumptions, which include that the 
functions $\mu_j(x)$ are twice continuously differentiable and  conditions on the bandwidth \cp{mack:82, fan:96}.

\begin{Theorem}
\label{thm: rate}
Let $h = h_n\rightarrow 0$ and $nh \rightarrow \infty.$  Under assumptions (A1)--(A5), for any interior point $x \in (0,1),$ the local Fr\'echet estimator satisfies 
$$d_F(\Sigma(x), \SLF(x)) = O_p\left(h^2 + (nh)^{-1/2} + q_n\right).$$
\end{Theorem}

\begin{Corollary}
\label{cor: rate}
Let $\hat{\mc{R}}^{\textrm{LF}}(x)$ and $\mc{R}(x)$ be the correlation matrices corresponding to $\SLF(x)$ and $\Sigma(x)$, respectively.  Under the assumptions of Theorem~\ref{thm: rate}, for any $x \in (0,1),$  $$d_F(\mc{R}(x), \hat{\mc{R}}^{\textrm{LF}}(t)) = O_p\left(h^2 + (nh)^{-1/2} +  q_n \right).$$
\end{Corollary}

From the theorem and corollary, the bandwidth choice which leads to the fastest rate of convergence is $h=h(n) \sim n^{-1/5}$ and the resulting $O_p$ rate is $n^{-2/5} + q_n$. The first term corresponds to the optimal rate of pointwise convergence for nonparametric estimation of regression functions that are twice continuously differentiable. A proof of Theorem~\ref{thm: rate} is provided below, while that of the corollary is straightforward and is omitted.

\begin{proof}[Proof of Theorem~\ref{thm: rate}]
First, assumptions (A1)--(A4) can be used to verify the conditions necessary to invoke Theorems 3 and 4 of \ci{mull:18:3}, with the exception that the space $\mathcal{S}_p$ is not bounded.  With minor alterations to the proofs, assumption (A3) implies that these continue to hold, since $\mathcal{S}_p$ is a subset of a Euclidean space, implying  $$d_F(\Sigma(x), \tilde{\Sigma}^{LF}(x)) = O_p(h^2 + (nh)^{-1/2}).$$  By the contraction property of the projection $\Pi_{\mc{S}_p}$, it follows that
$$
d_F(\tilde{\Sigma}^{LF}(x),\hat{\Sigma}^{LF}(x)) \leq \frac{1}{n}\left\lVert \sum_{i = 1}^n w_{in,1}(x)(C_i - \hat{C}_i)\right\rVert_F, 
$$
where $\lVert \cdot \rVert_F$ is the Frobenius norm. Assumption (A5) also implies that, for large enough $n$, we have
$$
\lVert C_i - \hat{C}_i\rVert_F \leq 4M_1\lVert \hat{\mu}(X_i) - \mu(X_i) \rVert.
$$
Under (A1), it can easily be shown that $|w_{in,1}(x)| = O_p(K_h(X_i - x)),$ where the $O_p$ term is uniform over $i$.  Hence, using (A5), we have
\begin{align*}
d_F(\tilde{\Sigma}^{LF}(x),\hat{\Sigma}^{LF}(x)) &\leq \frac{1}{n}\sum_{i = 1}^n |w_{in,1}(x)|\left\lVert C_i - \hat{C}_i\right\rVert_F \\
&= O_p\left(\frac{q_n}{n}\sum_{i = 1} K_h(X_i - x)\right) = O_p(q_n).
\end{align*}

\end{proof}


\begin{thebibliography}{61}
% BibTex style file: imsart-nameyear.bst, 2013-01-28
% Default style options (sort=1,type=nameyear).
% Used options (sort=1,type=nameyear).

\bibitem[\protect\citeauthoryear{Ash and Gardner}{1975}]{ash:75}
\begin{bbook}[author]
\bauthor{\bsnm{Ash},~\bfnm{Robert~B.}\binits{R.~B.}} \AND
  \bauthor{\bsnm{Gardner},~\bfnm{Melvin~F.}\binits{M.~F.}}
(\byear{1975}).
\btitle{Topics in {S}tochastic {P}rocesses}.
\bpublisher{Academic Press [Harcourt Brace Jovanovich Publishers]},
  \baddress{New York}.
\end{bbook}
\endbibitem

\bibitem[\protect\citeauthoryear{Bartzokis et~al.}{2010}]{d08}
\begin{barticle}[author]
\bauthor{\bsnm{Bartzokis},~\bfnm{G}\binits{G.}},
  \bauthor{\bsnm{Lu},~\bfnm{PH}\binits{P.}},
  \bauthor{\bsnm{Tingus},~\bfnm{K}\binits{K.}},
  \bauthor{\bsnm{Mendez},~\bfnm{MF}\binits{M.}},
  \bauthor{\bsnm{A},~\bfnm{Richard}\binits{R.}} \AND \bauthor{\bparticle{et}
  \bsnm{al.},~\bfnm{Peters~DG}\binits{P.~D.}}
(\byear{2010}).
\btitle{Lifespan trajectory of myelin integrity and maximum motor speed}.
\bjournal{Neurobiology of Aging}
\bvolume{31}
\bpages{1554--1562}.
\end{barticle}
\endbibitem

\bibitem[\protect\citeauthoryear{Bilenberg}{1999}]{d031}
\begin{barticle}[author]
\bauthor{\bsnm{Bilenberg},~\bfnm{Niels}\binits{N.}}
(\byear{1999}).
\btitle{The Child Behavior Checklist (CBCL) and related material:
  standardization and validation in Danish population based and clinically
  based samples}.
\bjournal{Acta Psychiatrica Scandinavica}
\bvolume{100}
\bpages{2--52}.
\end{barticle}
\endbibitem

\bibitem[\protect\citeauthoryear{Brody et~al.}{1987}]{d06}
\begin{barticle}[author]
\bauthor{\bsnm{Brody},~\bfnm{BA}\binits{B.}},
  \bauthor{\bsnm{Kinney},~\bfnm{HC}\binits{H.}},
  \bauthor{\bsnm{Kloman},~\bfnm{AS}\binits{A.}} \AND
  \bauthor{\bsnm{Gilles},~\bfnm{FH}\binits{F.}}
(\byear{1987}).
\btitle{Sequence of central nervous system myelination in human infancy. I. An
  autopsy study of myelination.}
\bjournal{J. Neuropathol. Exp. Neurol.}
\bvolume{46}
\bpages{283--301}.
\end{barticle}
\endbibitem

\bibitem[\protect\citeauthoryear{Carmichael et~al.}{2013}]{carm:13}
\begin{barticle}[author]
\bauthor{\bsnm{Carmichael},~\bfnm{Owen}\binits{O.}},
  \bauthor{\bsnm{Chen},~\bfnm{Jun}\binits{J.}},
  \bauthor{\bsnm{Paul},~\bfnm{Debashis}\binits{D.}} \AND
  \bauthor{\bsnm{Peng},~\bfnm{Jie}\binits{J.}}
(\byear{2013}).
\btitle{Diffusion tensor smoothing through weighted Karcher means}.
\bjournal{Electronic Journal of Statistics}
\bvolume{7}
\bpages{1913}.
\end{barticle}
\endbibitem

\bibitem[\protect\citeauthoryear{Casey, Galvan and Hare}{2005}]{d43}
\begin{barticle}[author]
\bauthor{\bsnm{Casey},~\bfnm{BJ}\binits{B.}},
  \bauthor{\bsnm{Galvan},~\bfnm{Adriana}\binits{A.}} \AND
  \bauthor{\bsnm{Hare},~\bfnm{Todd~A}\binits{T.~A.}}
(\byear{2005}).
\btitle{Changes in cerebral functional organization during cognitive
  development}.
\bjournal{Current opinion in neurobiology}
\bvolume{15}
\bpages{239--244}.
\end{barticle}
\endbibitem

\bibitem[\protect\citeauthoryear{Chiou and M\"{u}ller}{2016}]{mull:16:13}
\begin{barticle}[author]
\bauthor{\bsnm{Chiou},~\bfnm{Jeng-Min}\binits{J.-M.}} \AND
  \bauthor{\bsnm{M\"{u}ller},~\bfnm{H.~G.}\binits{H.~G.}}
(\byear{2016}).
\btitle{A pairwise interaction model for multivariate functional and
  longitudinal data}.
\bjournal{Biometrika}
\bvolume{103}
\bpages{377--396}.
\end{barticle}
\endbibitem

\bibitem[\protect\citeauthoryear{Chlebowski et~al.}{2013}]{d030}
\begin{barticle}[author]
\bauthor{\bsnm{Chlebowski},~\bfnm{Colby}\binits{C.}},
  \bauthor{\bsnm{Robins},~\bfnm{Diana~L}\binits{D.~L.}},
  \bauthor{\bsnm{Barton},~\bfnm{Marianne~L}\binits{M.~L.}} \AND
  \bauthor{\bsnm{Fein},~\bfnm{Deborah}\binits{D.}}
(\byear{2013}).
\btitle{Large-scale use of the modified checklist for autism in low-risk
  toddlers}.
\bjournal{Pediatrics}
\bvolume{131}
\bpages{e1121--e1127}.
\end{barticle}
\endbibitem

\bibitem[\protect\citeauthoryear{{\c{S}}ent{\"u}rk and Nguyen}{2011}]{sent:11}
\begin{barticle}[author]
\bauthor{\bsnm{{\c{S}}ent{\"u}rk},~\bfnm{Damla}\binits{D.}} \AND
  \bauthor{\bsnm{Nguyen},~\bfnm{Danh~V}\binits{D.~V.}}
(\byear{2011}).
\btitle{Varying coefficient models for sparse noise-contaminated longitudinal
  data}.
\bjournal{Statistica Sinica}
\bvolume{21}
\bpages{1831}.
\end{barticle}
\endbibitem

\bibitem[\protect\citeauthoryear{Dai et~al.}{2017a}]{mull:17:10}
\begin{barticle}[author]
\bauthor{\bsnm{Dai},~\bfnm{Xiongtao}\binits{X.}},
  \bauthor{\bsnm{M\"uller},~\bfnm{Hans-Georg}\binits{H.-G.}},
  \bauthor{\bsnm{Wang},~\bfnm{Jane-Ling}\binits{J.-L.}} \AND
  \bauthor{\bsnm{Deoni},~\bfnm{Sean~CL}\binits{S.~C.}}
(\byear{2017}a).
\btitle{Age-Dynamic Networks and Functional Correlation for Early White Matter
  Myelination}.
\bjournal{Preprint}.
\end{barticle}
\endbibitem

\bibitem[\protect\citeauthoryear{Dai et~al.}{2017b}]{mull:17:9}
\begin{barticle}[author]
\bauthor{\bsnm{Dai},~\bfnm{Xiongtao}\binits{X.}},
  \bauthor{\bsnm{Hadjipantelis},~\bfnm{Pantelis}\binits{P.}},
  \bauthor{\bsnm{Wang},~\bfnm{Jane-Ling}\binits{J.-L.}},
  \bauthor{\bsnm{Deoni},~\bfnm{Sean~CL}\binits{S.~C.}} \AND
  \bauthor{\bsnm{M\"uller},~\bfnm{Hans-Georg}\binits{H.-G.}}
(\byear{2017}b).
\btitle{Longitudinal Associations Between White Matter Maturation and Cognitive
  Development Across Early Childhood Xiongta}.
\bjournal{Preprint}.
\end{barticle}
\endbibitem

\bibitem[\protect\citeauthoryear{Deoni}{2007}]{d181}
\begin{barticle}[author]
\bauthor{\bsnm{Deoni},~\bfnm{Sean~CL}\binits{S.~C.}}
(\byear{2007}).
\btitle{High-resolution T1 mapping of the brain at 3T with driven equilibrium
  single pulse observation of T1 with high-speed incorporation of RF field
  inhomogeneities (DESPOT1-HIFI)}.
\bjournal{Journal of Magnetic Resonance Imaging}
\bvolume{26}
\bpages{1106--1111}.
\end{barticle}
\endbibitem

\bibitem[\protect\citeauthoryear{Deoni}{2011}]{d183}
\begin{barticle}[author]
\bauthor{\bsnm{Deoni},~\bfnm{Sean~CL}\binits{S.~C.}}
(\byear{2011}).
\btitle{Correction of main and transmit magnetic field (B0 and B1)
  inhomogeneity effects in multicomponent-driven equilibrium single-pulse
  observation of T1 and T2}.
\bjournal{Magnetic Resonance in Medicine}
\bvolume{65}
\bpages{1021--1035}.
\end{barticle}
\endbibitem

\bibitem[\protect\citeauthoryear{Deoni and Kolind}{2015}]{d187}
\begin{barticle}[author]
\bauthor{\bsnm{Deoni},~\bfnm{Sean~CL}\binits{S.~C.}} \AND
  \bauthor{\bsnm{Kolind},~\bfnm{Shannon~H}\binits{S.~H.}}
(\byear{2015}).
\btitle{Investigating the stability of mcDESPOT myelin water fraction values
  derived using a stochastic region contraction approach}.
\bjournal{Magnetic Resonance in Medicine}
\bvolume{73}
\bpages{161--169}.
\end{barticle}
\endbibitem

\bibitem[\protect\citeauthoryear{Deoni, Peters and Rutt}{2004}]{d186}
\begin{barticle}[author]
\bauthor{\bsnm{Deoni},~\bfnm{Sean~CL}\binits{S.~C.}},
  \bauthor{\bsnm{Peters},~\bfnm{Terry~M}\binits{T.~M.}} \AND
  \bauthor{\bsnm{Rutt},~\bfnm{Brian~K}\binits{B.~K.}}
(\byear{2004}).
\btitle{Determination of optimal angles for variable nutation proton magnetic
  spin-lattice, T1, and spin-spin, T2, relaxation times measurement}.
\bjournal{Magnetic Resonance in Medicine}
\bvolume{51}
\bpages{194--199}.
\end{barticle}
\endbibitem

\bibitem[\protect\citeauthoryear{Deoni, Rutt and Peters}{2003}]{d180}
\begin{barticle}[author]
\bauthor{\bsnm{Deoni},~\bfnm{Sean~CL}\binits{S.~C.}},
  \bauthor{\bsnm{Rutt},~\bfnm{Brian~K}\binits{B.~K.}} \AND
  \bauthor{\bsnm{Peters},~\bfnm{Terry~M}\binits{T.~M.}}
(\byear{2003}).
\btitle{Rapid combined T1 and T2 mapping using gradient recalled acquisition in
  the steady state}.
\bjournal{Magnetic Resonance in Medicine}
\bvolume{49}
\bpages{515--526}.
\end{barticle}
\endbibitem

\bibitem[\protect\citeauthoryear{Deoni, Rutt and Peters}{2006}]{d116}
\begin{barticle}[author]
\bauthor{\bsnm{Deoni},~\bfnm{Sean~CL}\binits{S.~C.}},
  \bauthor{\bsnm{Rutt},~\bfnm{Brian~K}\binits{B.~K.}} \AND
  \bauthor{\bsnm{Peters},~\bfnm{Terry~M}\binits{T.~M.}}
(\byear{2006}).
\btitle{Synthetic T 1-weighted brain image generation with incorporated coil
  intensity correction using DESPOT1}.
\bjournal{Magnetic resonance imaging}
\bvolume{24}
\bpages{1241--1248}.
\end{barticle}
\endbibitem

\bibitem[\protect\citeauthoryear{Deoni et~al.}{2004}]{d182}
\begin{barticle}[author]
\bauthor{\bsnm{Deoni},~\bfnm{Sean~CL}\binits{S.~C.}},
  \bauthor{\bsnm{Ward},~\bfnm{Heidi~A}\binits{H.~A.}},
  \bauthor{\bsnm{Peters},~\bfnm{Terry~M}\binits{T.~M.}} \AND
  \bauthor{\bsnm{Rutt},~\bfnm{Brian~K}\binits{B.~K.}}
(\byear{2004}).
\btitle{Rapid T2 estimation with phase-cycled variable nutation steady-state
  free precession}.
\bjournal{Magnetic resonance in medicine}
\bvolume{52}
\bpages{435--439}.
\end{barticle}
\endbibitem

\bibitem[\protect\citeauthoryear{Deoni et~al.}{2008a}]{d113}
\begin{barticle}[author]
\bauthor{\bsnm{Deoni},~\bfnm{Sean~CL}\binits{S.~C.}},
  \bauthor{\bsnm{Williams},~\bfnm{Steven~CR}\binits{S.~C.}},
  \bauthor{\bsnm{Jezzard},~\bfnm{Peter}\binits{P.}},
  \bauthor{\bsnm{Suckling},~\bfnm{John}\binits{J.}},
  \bauthor{\bsnm{Murphy},~\bfnm{Declan~GM}\binits{D.~G.}} \AND
  \bauthor{\bsnm{Jones},~\bfnm{Derek~K}\binits{D.~K.}}
(\byear{2008}a).
\btitle{Standardized structural magnetic resonance imaging in multicentre
  studies using quantitative T 1 and T 2 imaging at 1.5 T}.
\bjournal{Neuroimage}
\bvolume{40}
\bpages{662--671}.
\end{barticle}
\endbibitem

\bibitem[\protect\citeauthoryear{Deoni et~al.}{2008b}]{d161}
\begin{barticle}[author]
\bauthor{\bsnm{Deoni},~\bfnm{Sean~CL}\binits{S.~C.}},
  \bauthor{\bsnm{Rutt},~\bfnm{Brian~K}\binits{B.~K.}},
  \bauthor{\bsnm{Arun},~\bfnm{Tarunya}\binits{T.}},
  \bauthor{\bsnm{Pierpaoli},~\bfnm{Carlo}\binits{C.}} \AND
  \bauthor{\bsnm{Jones},~\bfnm{Derek~K}\binits{D.~K.}}
(\byear{2008}b).
\btitle{Gleaning multicomponent T1 and T2 information from steady-state imaging
  data}.
\bjournal{Magnetic Resonance in Medicine}
\bvolume{60}
\bpages{1372--1387}.
\end{barticle}
\endbibitem

\bibitem[\protect\citeauthoryear{Deoni et~al.}{2011}]{d165}
\begin{barticle}[author]
\bauthor{\bsnm{Deoni},~\bfnm{Sean~CL}\binits{S.~C.}},
  \bauthor{\bsnm{Mercure},~\bfnm{Evelyne}\binits{E.}},
  \bauthor{\bsnm{Blasi},~\bfnm{Anna}\binits{A.}},
  \bauthor{\bsnm{Gasston},~\bfnm{David}\binits{D.}},
  \bauthor{\bsnm{Thomson},~\bfnm{Alex}\binits{A.}},
  \bauthor{\bsnm{Johnson},~\bfnm{Mark}\binits{M.}},
  \bauthor{\bsnm{Williams},~\bfnm{Steven~CR}\binits{S.~C.}} \AND
  \bauthor{\bsnm{Murphy},~\bfnm{Declan~GM}\binits{D.~G.}}
(\byear{2011}).
\btitle{Mapping infant brain myelination with magnetic resonance imaging}.
\bjournal{The Journal of Neuroscience}
\bvolume{31}
\bpages{784--791}.
\end{barticle}
\endbibitem

\bibitem[\protect\citeauthoryear{Deoni et~al.}{2014}]{d142}
\begin{barticle}[author]
\bauthor{\bsnm{Deoni},~\bfnm{Sean~CL}\binits{S.~C.}},
  \bauthor{\bsnm{O’Muircheartaigh},~\bfnm{Jonathan}\binits{J.}},
  \bauthor{\bsnm{Elison},~\bfnm{Jed~T}\binits{J.~T.}},
  \bauthor{\bsnm{Walker},~\bfnm{Lindsay}\binits{L.}},
  \bauthor{\bsnm{Doernberg},~\bfnm{Ellen}\binits{E.}},
  \bauthor{\bsnm{Waskiewicz},~\bfnm{Nicole}\binits{N.}},
  \bauthor{\bsnm{Dirks},~\bfnm{Holly}\binits{H.}},
  \bauthor{\bsnm{Piryatinsky},~\bfnm{Irene}\binits{I.}},
  \bauthor{\bsnm{Dean~III},~\bfnm{Doug~C}\binits{D.~C.}} \AND
  \bauthor{\bsnm{Jumbe},~\bfnm{NL}\binits{N.}}
(\byear{2014}).
\btitle{White matter maturation profiles through early childhood predict
  general cognitive ability}.
\bjournal{Brain Structure and Function}
\bpages{1--15}.
\end{barticle}
\endbibitem

\bibitem[\protect\citeauthoryear{Deoni et~al.}{2015}]{d143}
\begin{barticle}[author]
\bauthor{\bsnm{Deoni},~\bfnm{Sean~CL}\binits{S.~C.}},
  \bauthor{\bsnm{Dean},~\bfnm{Douglas~C}\binits{D.~C.}},
  \bauthor{\bsnm{Remer},~\bfnm{Justin}\binits{J.}},
  \bauthor{\bsnm{Dirks},~\bfnm{Holly}\binits{H.}} \AND
  \bauthor{\bsnm{O’Muircheartaigh},~\bfnm{Jonathan}\binits{J.}}
(\byear{2015}).
\btitle{Cortical maturation and myelination in healthy toddlers and young
  children}.
\bjournal{NeuroImage}
\bvolume{115}
\bpages{147--161}.
\end{barticle}
\endbibitem

\bibitem[\protect\citeauthoryear{Erus et~al.}{2014}]{erus:14}
\begin{barticle}[author]
\bauthor{\bsnm{Erus},~\bfnm{Guray}\binits{G.}},
  \bauthor{\bsnm{Battapady},~\bfnm{Harsha}\binits{H.}},
  \bauthor{\bsnm{Satterthwaite},~\bfnm{Theodore~D}\binits{T.~D.}},
  \bauthor{\bsnm{Hakonarson},~\bfnm{Hakon}\binits{H.}},
  \bauthor{\bsnm{Gur},~\bfnm{Raquel~E}\binits{R.~E.}},
  \bauthor{\bsnm{Davatzikos},~\bfnm{Christos}\binits{C.}} \AND
  \bauthor{\bsnm{Gur},~\bfnm{Ruben~C}\binits{R.~C.}}
(\byear{2014}).
\btitle{Imaging patterns of brain development and their relationship to
  cognition}.
\bjournal{Cerebral Cortex}
\bvolume{25}
\bpages{1676--1684}.
\end{barticle}
\endbibitem

\bibitem[\protect\citeauthoryear{Fan and Gijbels}{1996}]{fan:96}
\begin{bbook}[author]
\bauthor{\bsnm{Fan},~\bfnm{J.}\binits{J.}} \AND
  \bauthor{\bsnm{Gijbels},~\bfnm{I.}\binits{I.}}
(\byear{1996}).
\btitle{Local Polynomial Modelling and its Applications}.
\bpublisher{Chapman \& Hall}, \baddress{London}.
\bmrnumber{MR1383587 (97f:62063)}
\end{bbook}
\endbibitem

\bibitem[\protect\citeauthoryear{Fields}{2005}]{d09}
\begin{barticle}[author]
\bauthor{\bsnm{Fields},~\bfnm{RD}\binits{R.}}
(\byear{2005}).
\btitle{Myelination: an overlooked mechanism of synaptic plasticity?}
\bjournal{Neuroscientist. SAGE Publications}
\bvolume{11}
\bpages{528--531}.
\end{barticle}
\endbibitem

\bibitem[\protect\citeauthoryear{Fields}{2008}]{d012}
\begin{barticle}[author]
\bauthor{\bsnm{Fields},~\bfnm{R~Douglas}\binits{R.~D.}}
(\byear{2008}).
\btitle{White matter in learning, cognition and psychiatric disorders}.
\bjournal{Trends in neurosciences}
\bvolume{31}
\bpages{361--370}.
\end{barticle}
\endbibitem

\bibitem[\protect\citeauthoryear{Flynn et~al.}{2003}]{d015}
\begin{barticle}[author]
\bauthor{\bsnm{Flynn},~\bfnm{SW}\binits{S.}},
  \bauthor{\bsnm{Lang},~\bfnm{DJ}\binits{D.}},
  \bauthor{\bsnm{Mackay},~\bfnm{AL}\binits{A.}},
  \bauthor{\bsnm{Goghari},~\bfnm{V}\binits{V.}},
  \bauthor{\bsnm{Vavasour},~\bfnm{IM}\binits{I.}},
  \bauthor{\bsnm{Whittall},~\bfnm{KP}\binits{K.}},
  \bauthor{\bsnm{Smith},~\bfnm{GN}\binits{G.}},
  \bauthor{\bsnm{Arango},~\bfnm{V}\binits{V.}},
  \bauthor{\bsnm{Mann},~\bfnm{JJ}\binits{J.}},
  \bauthor{\bsnm{Dwork},~\bfnm{AJ}\binits{A.}} \betal{et~al.}
(\byear{2003}).
\btitle{Abnormalities of myelination in schizophrenia detected in vivo with
  MRI, and post-mortem with analysis of oligodendrocyte proteins}.
\bjournal{Molecular psychiatry}
\bvolume{8}
\bpages{811--820}.
\end{barticle}
\endbibitem

\bibitem[\protect\citeauthoryear{Fornari et~al.}{2007}]{d45}
\begin{barticle}[author]
\bauthor{\bsnm{Fornari},~\bfnm{Eleonora}\binits{E.}},
  \bauthor{\bsnm{Knyazeva},~\bfnm{Maria~G}\binits{M.~G.}},
  \bauthor{\bsnm{Meuli},~\bfnm{Reto}\binits{R.}} \AND
  \bauthor{\bsnm{Maeder},~\bfnm{Philippe}\binits{P.}}
(\byear{2007}).
\btitle{Myelination shapes functional activity in the developing brain}.
\bjournal{Neuroimage}
\bvolume{38}
\bpages{511--518}.
\end{barticle}
\endbibitem

\bibitem[\protect\citeauthoryear{Fryer et~al.}{2008}]{frye:08}
\begin{barticle}[author]
\bauthor{\bsnm{Fryer},~\bfnm{Susanna~L}\binits{S.~L.}},
  \bauthor{\bsnm{Frank},~\bfnm{Lawrence~R}\binits{L.~R.}},
  \bauthor{\bsnm{Spadoni},~\bfnm{Andrea~D}\binits{A.~D.}},
  \bauthor{\bsnm{Theilmann},~\bfnm{Rebecca~J}\binits{R.~J.}},
  \bauthor{\bsnm{Nagel},~\bfnm{Bonnie~J}\binits{B.~J.}},
  \bauthor{\bsnm{Schweinsburg},~\bfnm{Alecia~D}\binits{A.~D.}} \AND
  \bauthor{\bsnm{Tapert},~\bfnm{Susan~F}\binits{S.~F.}}
(\byear{2008}).
\btitle{Microstructural integrity of the corpus callosum linked with
  neuropsychological performance in adolescents}.
\bjournal{Brain and Cognition}
\bvolume{67}
\bpages{225--233}.
\end{barticle}
\endbibitem

\bibitem[\protect\citeauthoryear{Grydeland et~al.}{2013}]{d63}
\begin{barticle}[author]
\bauthor{\bsnm{Grydeland},~\bfnm{H{\aa}kon}\binits{H.}},
  \bauthor{\bsnm{Walhovd},~\bfnm{Kristine~B}\binits{K.~B.}},
  \bauthor{\bsnm{Tamnes},~\bfnm{Christian~K}\binits{C.~K.}},
  \bauthor{\bsnm{Westlye},~\bfnm{Lars~T}\binits{L.~T.}} \AND
  \bauthor{\bsnm{Fjell},~\bfnm{Anders~M}\binits{A.~M.}}
(\byear{2013}).
\btitle{Intracortical myelin links with performance variability across the
  human lifespan: results from T1-and T2-weighted MRI myelin mapping and
  diffusion tensor imaging}.
\bjournal{The Journal of Neuroscience}
\bvolume{33}
\bpages{18618--18630}.
\end{barticle}
\endbibitem

\bibitem[\protect\citeauthoryear{Hagmann et~al.}{2010}]{d62}
\begin{barticle}[author]
\bauthor{\bsnm{Hagmann},~\bfnm{Patric}\binits{P.}},
  \bauthor{\bsnm{Sporns},~\bfnm{Olaf}\binits{O.}},
  \bauthor{\bsnm{Madan},~\bfnm{Neel}\binits{N.}},
  \bauthor{\bsnm{Cammoun},~\bfnm{Leila}\binits{L.}},
  \bauthor{\bsnm{Pienaar},~\bfnm{Rudolph}\binits{R.}},
  \bauthor{\bsnm{Wedeen},~\bfnm{Van~Jay}\binits{V.~J.}},
  \bauthor{\bsnm{Meuli},~\bfnm{Reto}\binits{R.}},
  \bauthor{\bsnm{Thiran},~\bfnm{J-P}\binits{J.-P.}} \AND
  \bauthor{\bsnm{Grant},~\bfnm{PE}\binits{P.}}
(\byear{2010}).
\btitle{White matter maturation reshapes structural connectivity in the late
  developing human brain}.
\bjournal{Proceedings of the National Academy of Sciences}
\bvolume{107}
\bpages{19067--19072}.
\end{barticle}
\endbibitem

\bibitem[\protect\citeauthoryear{Hsing and Eubank}{2015}]{hsin:15}
\begin{bbook}[author]
\bauthor{\bsnm{Hsing},~\bfnm{Tailen}\binits{T.}} \AND
  \bauthor{\bsnm{Eubank},~\bfnm{Randall}\binits{R.}}
(\byear{2015}).
\btitle{Theoretical Foundations of Functional Data Analysis, with an
  Introduction to Linear Operators}.
\bpublisher{John Wiley \& Sons}.
\end{bbook}
\endbibitem

\bibitem[\protect\citeauthoryear{Ishibashi et~al.}{2006}]{d010}
\begin{barticle}[author]
\bauthor{\bsnm{Ishibashi},~\bfnm{T}\binits{T.}},
  \bauthor{\bsnm{Dakin},~\bfnm{KA}\binits{K.}},
  \bauthor{\bsnm{Stevens},~\bfnm{B}\binits{B.}},
  \bauthor{\bsnm{Lee},~\bfnm{PR}\binits{P.}},
  \bauthor{\bsnm{Kozlov},~\bfnm{SV}\binits{S.}} \AND
  \bauthor{\bsnm{Stewart},~\bfnm{CL~et~al.}\binits{C.~e.~a.}}
(\byear{2006}).
\btitle{Astrocytes promote myelination in response to electrical impulses}.
\bjournal{Neuron}
\bvolume{49}
\bpages{823--832}.
\end{barticle}
\endbibitem

\bibitem[\protect\citeauthoryear{Knickmeyer et~al.}{2008}]{d30}
\begin{barticle}[author]
\bauthor{\bsnm{Knickmeyer},~\bfnm{Rebecca~C}\binits{R.~C.}},
  \bauthor{\bsnm{Gouttard},~\bfnm{Sylvain}\binits{S.}},
  \bauthor{\bsnm{Kang},~\bfnm{Chaeryon}\binits{C.}},
  \bauthor{\bsnm{Evans},~\bfnm{Dianne}\binits{D.}},
  \bauthor{\bsnm{Wilber},~\bfnm{Kathy}\binits{K.}},
  \bauthor{\bsnm{Smith},~\bfnm{J~Keith}\binits{J.~K.}},
  \bauthor{\bsnm{Hamer},~\bfnm{Robert~M}\binits{R.~M.}},
  \bauthor{\bsnm{Lin},~\bfnm{Weili}\binits{W.}},
  \bauthor{\bsnm{Gerig},~\bfnm{Guido}\binits{G.}} \AND
  \bauthor{\bsnm{Gilmore},~\bfnm{John~H}\binits{J.~H.}}
(\byear{2008}).
\btitle{A structural MRI study of human brain development from birth to 2
  years}.
\bjournal{The Journal of Neuroscience}
\bvolume{28}
\bpages{12176--12182}.
\end{barticle}
\endbibitem

\bibitem[\protect\citeauthoryear{Lang et~al.}{2003}]{d011}
\begin{barticle}[author]
\bauthor{\bsnm{Lang},~\bfnm{DJM}\binits{D.}},
  \bauthor{\bsnm{Yip},~\bfnm{E}\binits{E.}},
  \bauthor{\bsnm{MacKay},~\bfnm{AL}\binits{A.}},
  \bauthor{\bsnm{Thornton},~\bfnm{AE}\binits{A.}},
  \bauthor{\bsnm{Vila-Rodriguez},~\bfnm{F}\binits{F.}} \AND
  \bauthor{\bsnm{MacEwan},~\bfnm{GW~et~al.}\binits{G.~e.~a.}}
(\byear{2003}).
\btitle{Abnormalities of myelination in schizophrenia detected in vivo with
  MRI, and post-mortem with analysis of oligodendrocyte proteins}.
\bjournal{Mol Psychiatry}
\bvolume{8}
\bpages{811--820}.
\end{barticle}
\endbibitem

\bibitem[\protect\citeauthoryear{Lebel and Beaulieu}{2011}]{d16}
\begin{barticle}[author]
\bauthor{\bsnm{Lebel},~\bfnm{Catherine}\binits{C.}} \AND
  \bauthor{\bsnm{Beaulieu},~\bfnm{Christian}\binits{C.}}
(\byear{2011}).
\btitle{Longitudinal development of human brain wiring continues from childhood
  into adulthood}.
\bjournal{The Journal of Neuroscience}
\bvolume{31}
\bpages{10937--10947}.
\end{barticle}
\endbibitem

\bibitem[\protect\citeauthoryear{Lenroot and Giedd}{2006}]{lenr:06}
\begin{barticle}[author]
\bauthor{\bsnm{Lenroot},~\bfnm{Rhoshel~K}\binits{R.~K.}} \AND
  \bauthor{\bsnm{Giedd},~\bfnm{Jay~N}\binits{J.~N.}}
(\byear{2006}).
\btitle{Brain development in children and adolescents: insights from anatomical
  magnetic resonance imaging}.
\bjournal{Neuroscience \& Biobehavioral Reviews}
\bvolume{30}
\bpages{718--729}.
\end{barticle}
\endbibitem

\bibitem[\protect\citeauthoryear{Mack and Silverman}{1982}]{mack:82}
\begin{barticle}[author]
\bauthor{\bsnm{Mack},~\bfnm{Y.~P.}\binits{Y.~P.}} \AND
  \bauthor{\bsnm{Silverman},~\bfnm{B.~W.}\binits{B.~W.}}
(\byear{1982}).
\btitle{Weak and Strong Uniform Consistency of Kernel Regression Estimates}.
\bjournal{Zeitschrift f\"{u}r Wahrscheinlichkeitstheorie und verwandte Gebiete}
\bvolume{61}
\bpages{405--415}.
\end{barticle}
\endbibitem

\bibitem[\protect\citeauthoryear{Miller et~al.}{2012}]{d47}
\begin{barticle}[author]
\bauthor{\bsnm{Miller},~\bfnm{Daniel~J}\binits{D.~J.}},
  \bauthor{\bsnm{Duka},~\bfnm{Tetyana}\binits{T.}},
  \bauthor{\bsnm{Stimpson},~\bfnm{Cheryl~D}\binits{C.~D.}},
  \bauthor{\bsnm{Schapiro},~\bfnm{Steven~J}\binits{S.~J.}},
  \bauthor{\bsnm{Baze},~\bfnm{Wallace~B}\binits{W.~B.}},
  \bauthor{\bsnm{McArthur},~\bfnm{Mark~J}\binits{M.~J.}},
  \bauthor{\bsnm{Fobbs},~\bfnm{Archibald~J}\binits{A.~J.}},
  \bauthor{\bsnm{Sousa},~\bfnm{Andr{\'e}~MM}\binits{A.~M.}},
  \bauthor{\bsnm{{\v{S}}estan},~\bfnm{Nenad}\binits{N.}},
  \bauthor{\bsnm{Wildman},~\bfnm{Derek~E}\binits{D.~E.}} \betal{et~al.}
(\byear{2012}).
\btitle{Prolonged myelination in human neocortical evolution}.
\bjournal{Proceedings of the National Academy of Sciences}
\bvolume{109}
\bpages{16480--16485}.
\end{barticle}
\endbibitem

\bibitem[\protect\citeauthoryear{Mullen et~al.}{1995}]{d124}
\begin{bbook}[author]
\bauthor{\bsnm{Mullen},~\bfnm{Eileen~M}\binits{E.~M.}} \betal{et~al.}
(\byear{1995}).
\btitle{Mullen scales of early learning}.
\bpublisher{AGS Circle Pines, MN}.
\end{bbook}
\endbibitem

\bibitem[\protect\citeauthoryear{Nagy, Westerberg and Klingberg}{2004}]{d013}
\begin{barticle}[author]
\bauthor{\bsnm{Nagy},~\bfnm{Zoltan}\binits{Z.}},
  \bauthor{\bsnm{Westerberg},~\bfnm{Helena}\binits{H.}} \AND
  \bauthor{\bsnm{Klingberg},~\bfnm{Torkel}\binits{T.}}
(\byear{2004}).
\btitle{Maturation of white matter is associated with the development of
  cognitive functions during childhood}.
\bjournal{Journal of cognitive neuroscience}
\bvolume{16}
\bpages{1227--1233}.
\end{barticle}
\endbibitem

\bibitem[\protect\citeauthoryear{Paus et~al.}{2001}]{d07}
\begin{barticle}[author]
\bauthor{\bsnm{Paus},~\bfnm{T}\binits{T.}},
  \bauthor{\bsnm{Collins},~\bfnm{DL}\binits{D.}},
  \bauthor{\bsnm{Evans},~\bfnm{AC}\binits{A.}},
  \bauthor{\bsnm{Leonard},~\bfnm{G}\binits{G.}},
  \bauthor{\bsnm{Pike},~\bfnm{B}\binits{B.}} \AND
  \bauthor{\bsnm{Zijdenbos},~\bfnm{A}\binits{A.}}
(\byear{2001}).
\btitle{Maturation of white matter in the human brain: a review of magnetic
  resonance studies}.
\bjournal{Brain Research Bulletin}
\bvolume{54}
\bpages{255--266}.
\end{barticle}
\endbibitem

\bibitem[\protect\citeauthoryear{Petersen and M{\"u}ller}{2016}]{mull:16:2}
\begin{barticle}[author]
\bauthor{\bsnm{Petersen},~\bfnm{Alexander}\binits{A.}} \AND
  \bauthor{\bsnm{M{\"u}ller},~\bfnm{Hans-Georg}\binits{H.-G.}}
(\byear{2016}).
\btitle{Fr\'echet integration and adaptive metric selection for interpretable
  covariances of multivariate functional data}.
\bjournal{Biometrika}
\bvolume{103}
\bpages{103--120}.
\end{barticle}
\endbibitem

\bibitem[\protect\citeauthoryear{Petersen and M\"{u}ller}{2018}]{mull:18:3}
\begin{barticle}[author]
\bauthor{\bsnm{Petersen},~\bfnm{Alexander}\binits{A.}} \AND
  \bauthor{\bsnm{M\"{u}ller},~\bfnm{Hans-Georg}\binits{H.-G.}}
(\byear{2018}).
\btitle{Fr\'echet regression for random objects with {E}uclidean predictors}.
\bjournal{Annals of Statistics, to appear (arXiv preprint arXiv:1608.03012)}.
\end{barticle}
\endbibitem

\bibitem[\protect\citeauthoryear{Petersen, M\"{u}ller and
  Deoni}{2018}]{pete:18}
\begin{barticle}[author]
\bauthor{\bsnm{Petersen},~\bfnm{Alexander}\binits{A.}},
  \bauthor{\bsnm{M\"{u}ller},~\bfnm{Hans-Georg}\binits{H.-G.}} \AND
  \bauthor{\bsnm{Deoni},~\bfnm{Sean}\binits{S.}}
(\byear{2018}).
\btitle{Supplement to ``Fr\'echet Estimation of Time-Varying Covariance
  Matrices From Sparse Data, With Application to the Regional Co-Evolution of
  Myelination in the Developing Brain''}.
\end{barticle}
\endbibitem

\bibitem[\protect\citeauthoryear{Pigoli et~al.}{2014}]{pigo:14}
\begin{barticle}[author]
\bauthor{\bsnm{Pigoli},~\bfnm{Davide}\binits{D.}},
  \bauthor{\bsnm{Aston},~\bfnm{John~AD}\binits{J.~A.}},
  \bauthor{\bsnm{Dryden},~\bfnm{Ian~L}\binits{I.~L.}} \AND
  \bauthor{\bsnm{Secchi},~\bfnm{Piercesare}\binits{P.}}
(\byear{2014}).
\btitle{Distances and inference for covariance operators}.
\bjournal{Biometrika}
\bvolume{101}
\bpages{409--422}.
\end{barticle}
\endbibitem

\bibitem[\protect\citeauthoryear{Rodier}{1995}]{d014}
\begin{barticle}[author]
\bauthor{\bsnm{Rodier},~\bfnm{Patricia~M}\binits{P.~M.}}
(\byear{1995}).
\btitle{Developing brain as a target of toxicity.}
\bjournal{Environmental Health Perspectives}
\bvolume{103}
\bpages{73}.
\end{barticle}
\endbibitem

\bibitem[\protect\citeauthoryear{Sent\"{u}rk and M\"{u}ller}{2010}]{mull:10:6}
\begin{barticle}[author]
\bauthor{\bsnm{Sent\"{u}rk},~\bfnm{Damla}\binits{D.}} \AND
  \bauthor{\bsnm{M\"{u}ller},~\bfnm{Hans-Georg}\binits{H.-G.}}
(\byear{2010}).
\btitle{Functional varying coefficient models for longitudinal data}.
\bjournal{Journal of the American Statistical Association}
\bvolume{105}
\bpages{1256-1264}.
\end{barticle}
\endbibitem

\bibitem[\protect\citeauthoryear{Shafee, Buckner and Fischl}{2015}]{d68}
\begin{barticle}[author]
\bauthor{\bsnm{Shafee},~\bfnm{Rebecca}\binits{R.}},
  \bauthor{\bsnm{Buckner},~\bfnm{Randy~L}\binits{R.~L.}} \AND
  \bauthor{\bsnm{Fischl},~\bfnm{Bruce}\binits{B.}}
(\byear{2015}).
\btitle{Gray matter myelination of 1555 human brains using partial volume
  corrected MRI images}.
\bjournal{Neuroimage}
\bvolume{105}
\bpages{473--485}.
\end{barticle}
\endbibitem

\bibitem[\protect\citeauthoryear{Shaw et~al.}{2006}]{d61}
\begin{barticle}[author]
\bauthor{\bsnm{Shaw},~\bfnm{Philip}\binits{P.}},
  \bauthor{\bsnm{Greenstein},~\bfnm{Deanna}\binits{D.}},
  \bauthor{\bsnm{Lerch},~\bfnm{Jason}\binits{J.}},
  \bauthor{\bsnm{Clasen},~\bfnm{Liv}\binits{L.}},
  \bauthor{\bsnm{Lenroot},~\bfnm{Rhoshel}\binits{R.}},
  \bauthor{\bsnm{Gogtay},~\bfnm{N}\binits{N.}},
  \bauthor{\bsnm{Evans},~\bfnm{Alan}\binits{A.}},
  \bauthor{\bsnm{Rapoport},~\bfnm{J}\binits{J.}} \AND
  \bauthor{\bsnm{Giedd},~\bfnm{J}\binits{J.}}
(\byear{2006}).
\btitle{Intellectual ability and cortical development in children and
  adolescents}.
\bjournal{Nature}
\bvolume{440}
\bpages{676--679}.
\end{barticle}
\endbibitem

\bibitem[\protect\citeauthoryear{Shaw et~al.}{2008}]{d28}
\begin{barticle}[author]
\bauthor{\bsnm{Shaw},~\bfnm{Philip}\binits{P.}},
  \bauthor{\bsnm{Kabani},~\bfnm{Noor~J}\binits{N.~J.}},
  \bauthor{\bsnm{Lerch},~\bfnm{Jason~P}\binits{J.~P.}},
  \bauthor{\bsnm{Eckstrand},~\bfnm{Kristen}\binits{K.}},
  \bauthor{\bsnm{Lenroot},~\bfnm{Rhoshel}\binits{R.}},
  \bauthor{\bsnm{Gogtay},~\bfnm{Nitin}\binits{N.}},
  \bauthor{\bsnm{Greenstein},~\bfnm{Deanna}\binits{D.}},
  \bauthor{\bsnm{Clasen},~\bfnm{Liv}\binits{L.}},
  \bauthor{\bsnm{Evans},~\bfnm{Alan}\binits{A.}},
  \bauthor{\bsnm{Rapoport},~\bfnm{Judith~L}\binits{J.~L.}} \betal{et~al.}
(\byear{2008}).
\btitle{Neurodevelopmental trajectories of the human cerebral cortex}.
\bjournal{The Journal of Neuroscience}
\bvolume{28}
\bpages{3586--3594}.
\end{barticle}
\endbibitem

\bibitem[\protect\citeauthoryear{Slater}{1997}]{slat:97}
\begin{barticle}[author]
\bauthor{\bsnm{Slater},~\bfnm{Alan}\binits{A.}}
(\byear{1997}).
\btitle{Can measures of infant habituation predict later intellectual ability?}
\bjournal{Archives of disease in childhood}
\bvolume{77}
\bpages{474--476}.
\end{barticle}
\endbibitem

\bibitem[\protect\citeauthoryear{Smyser et~al.}{2016}]{smys:16}
\begin{barticle}[author]
\bauthor{\bsnm{Smyser},~\bfnm{Tara~A}\binits{T.~A.}},
  \bauthor{\bsnm{Smyser},~\bfnm{Christopher~D}\binits{C.~D.}},
  \bauthor{\bsnm{Rogers},~\bfnm{Cynthia~E}\binits{C.~E.}},
  \bauthor{\bsnm{Gillespie},~\bfnm{Sarah~K}\binits{S.~K.}},
  \bauthor{\bsnm{Inder},~\bfnm{Terrie~E}\binits{T.~E.}} \AND
  \bauthor{\bsnm{Neil},~\bfnm{Jeffrey~J}\binits{J.~J.}}
(\byear{2016}).
\btitle{Cortical gray and adjacent white matter demonstrate synchronous
  maturation in very preterm infants}.
\bjournal{Cerebral Cortex}
\bvolume{26}
\bpages{3370--3378}.
\end{barticle}
\endbibitem

\bibitem[\protect\citeauthoryear{Stiles and Jernigan}{2010}]{d04}
\begin{barticle}[author]
\bauthor{\bsnm{Stiles},~\bfnm{J}\binits{J.}} \AND
  \bauthor{\bsnm{Jernigan},~\bfnm{TL}\binits{T.}}
(\byear{2010}).
\btitle{The basics of brain development}.
\bjournal{Neuropsychol. Review}
\bvolume{20}
\bpages{327--348}.
\end{barticle}
\endbibitem

\bibitem[\protect\citeauthoryear{Tavakoli et~al.}{2016}]{tava:16}
\begin{barticle}[author]
\bauthor{\bsnm{Tavakoli},~\bfnm{Shahin}\binits{S.}},
  \bauthor{\bsnm{Pigoli},~\bfnm{Davide}\binits{D.}},
  \bauthor{\bsnm{Aston},~\bfnm{John~AD}\binits{J.~A.}} \AND
  \bauthor{\bsnm{Coleman},~\bfnm{John}\binits{J.}}
(\byear{2016}).
\btitle{Spatial modeling of Object Data: Analysing dialect sound variations
  across the {UK}}.
\bjournal{arXiv preprint arXiv:1610.10040}.
\end{barticle}
\endbibitem

\bibitem[\protect\citeauthoryear{Wang, Chiou and M\"uller}{2016}]{mull:16:3}
\begin{barticle}[author]
\bauthor{\bsnm{Wang},~\bfnm{Jane-Ling}\binits{J.-L.}},
  \bauthor{\bsnm{Chiou},~\bfnm{Jeng-Min}\binits{J.-M.}} \AND
  \bauthor{\bsnm{M\"uller},~\bfnm{Hans-Georg}\binits{H.-G.}}
(\byear{2016}).
\btitle{Functional Data Analysis}.
\bjournal{Annual Review of Statistics and its Application}
\bvolume{3}
\bpages{257--295}.
\end{barticle}
\endbibitem

\bibitem[\protect\citeauthoryear{Wolff et~al.}{2014}]{d12}
\begin{barticle}[author]
\bauthor{\bsnm{Wolff},~\bfnm{Jason~J}\binits{J.~J.}},
  \bauthor{\bsnm{Gu},~\bfnm{Hongbin}\binits{H.}},
  \bauthor{\bsnm{Gerig},~\bfnm{Guido}\binits{G.}},
  \bauthor{\bsnm{Elison},~\bfnm{Jed~T}\binits{J.~T.}},
  \bauthor{\bsnm{Styner},~\bfnm{Martin}\binits{M.}},
  \bauthor{\bsnm{Gouttard},~\bfnm{Sylvain}\binits{S.}},
  \bauthor{\bsnm{Botteron},~\bfnm{Kelly~N}\binits{K.~N.}},
  \bauthor{\bsnm{Dager},~\bfnm{Stephen~R}\binits{S.~R.}},
  \bauthor{\bsnm{Dawson},~\bfnm{Geraldine}\binits{G.}},
  \bauthor{\bsnm{Estes},~\bfnm{Annette~M}\binits{A.~M.}} \betal{et~al.}
(\byear{2014}).
\btitle{Differences in white matter fiber tract development present from 6 to
  24 months in infants with autism}.
\bjournal{American Journal of Psychiatry}.
\end{barticle}
\endbibitem

\bibitem[\protect\citeauthoryear{Xiao et~al.}{2014}]{d016}
\begin{barticle}[author]
\bauthor{\bsnm{Xiao},~\bfnm{Zhou}\binits{Z.}},
  \bauthor{\bsnm{Qiu},~\bfnm{Ting}\binits{T.}},
  \bauthor{\bsnm{Ke},~\bfnm{Xiaoyan}\binits{X.}},
  \bauthor{\bsnm{Xiao},~\bfnm{Xiang}\binits{X.}},
  \bauthor{\bsnm{Xiao},~\bfnm{Ting}\binits{T.}},
  \bauthor{\bsnm{Liang},~\bfnm{Fengjing}\binits{F.}},
  \bauthor{\bsnm{Zou},~\bfnm{Bing}\binits{B.}},
  \bauthor{\bsnm{Huang},~\bfnm{Haiqing}\binits{H.}},
  \bauthor{\bsnm{Fang},~\bfnm{Hui}\binits{H.}},
  \bauthor{\bsnm{Chu},~\bfnm{Kangkang}\binits{K.}} \betal{et~al.}
(\byear{2014}).
\btitle{Autism spectrum disorder as early neurodevelopmental disorder: evidence
  from the brain imaging abnormalities in 2--3 years old toddlers}.
\bjournal{Journal of autism and developmental disorders}
\bvolume{44}
\bpages{1633--1640}.
\end{barticle}
\endbibitem

\bibitem[\protect\citeauthoryear{Yuan et~al.}{2012}]{yuan:12}
\begin{barticle}[author]
\bauthor{\bsnm{Yuan},~\bfnm{Ying}\binits{Y.}},
  \bauthor{\bsnm{Zhu},~\bfnm{Hongtu}\binits{H.}},
  \bauthor{\bsnm{Lin},~\bfnm{Weili}\binits{W.}} \AND
  \bauthor{\bsnm{Marron},~\bfnm{JS}\binits{J.}}
(\byear{2012}).
\btitle{Local polynomial regression for symmetric positive definite matrices}.
\bjournal{Journal of the Royal Statistical Society: Series B (Statistical
  Methodology)}
\bvolume{74}
\bpages{697--719}.
\end{barticle}
\endbibitem

\bibitem[\protect\citeauthoryear{Zhou, Huang and Carroll}{2008}]{zhou:08}
\begin{barticle}[author]
\bauthor{\bsnm{Zhou},~\bfnm{L.}\binits{L.}},
  \bauthor{\bsnm{Huang},~\bfnm{J.~Z.}\binits{J.~Z.}} \AND
  \bauthor{\bsnm{Carroll},~\bfnm{R.~J.}\binits{R.~J.}}
(\byear{2008}).
\btitle{Joint modelling of paired sparse functional data using principal
  components}.
\bjournal{Biometrika}
\bvolume{95}
\bpages{601-619}.
\end{barticle}
\endbibitem

\end{thebibliography}
\end{document}